\documentclass[journal]{IEEEtran}

\usepackage{amsmath}
\usepackage{amsfonts}
\usepackage{amssymb}
\usepackage{amsthm}
\usepackage{tabularx}
\usepackage{mathrsfs} 

\usepackage{url}
\usepackage{hyperref}
\newtheorem{theorem}{Theorem}
\newtheorem{remark}{Remark}
\newtheorem{definition}{Definition}

\newtheorem{corollary}{Corollary}

\usepackage{stfloats}
\usepackage{float}
\usepackage{graphicx}
\usepackage{cite}
\usepackage{xcolor}
\usepackage{subfigure}
\renewcommand{\vec}[1]{\mathbf{#1}}

\makeatletter
\def\blfootnote{\xdef\@thefnmark{}\@footnotetext}
\makeatother
\begin{document}
	
	\title{\huge{Performance Analysis of SWIPT Relay Networks over Arbitrary Dependent Fading Channels}} 
	\author{Farshad~Rostami~Ghadi\IEEEmembership{}~and  F.~Javier~L\'opez-Mart\'inez\IEEEmembership{}}
	\maketitle
	\begin{abstract}
		%{%\textcolor{blue}
		%}
		In this paper, we investigate the impact of fading channel correlation on the performance of dual-hop decode-and-forward (DF) simultaneous wireless information and power transfer (SWIPT) relay networks. More specifically, by considering the power splitting-based relaying (PSR) protocol for the energy harvesting (EH) process, we quantify the effect of positive and negative dependency between the source-to-relay ($SR$) and relay-to-destination ($RD$) links on key performance metrics such as ergodic capacity and outage probability. To this end, we first represent general formulations for the cumulative distribution function (CDF) %and the probability density function (PDF) 
		of the product of two arbitrary random variables, exploiting copula theory. This is used to derive the closed-form expressions of the ergodic capacity and outage probability in a SWIPT relay network under correlated Nakagami-$m$ fading channels. Monte-Carlo simulation results are provided to validate the correctness of the developed analytical results, showing that the system performance significantly improves under positive dependence in the $SR$-$RD$ links, compared to the case of negative dependence and independent links. Results further demonstrate that the efficiency of the ergodic capacity and outage probability ameliorates as the fading severity reduces for the PSR protocol.  
	\end{abstract}
	\begin{IEEEkeywords}
		Relay network, SWIPT, ergodic capacity, outage probability, Nakagami-$m$ model, energy harvesting, fading correlation, copula theory
	\end{IEEEkeywords}%\vspace{-3.5ex}
	\maketitle
	\blfootnote{\noindent Manuscript received June xx, 2022; revised XXX.} %This work has been funded in part by the Spanish Government and the European Fund for Regional Development FEDER (project  TEC2017-87913-R) and by Junta de Andalucia (project P18-RT-3175). The review of this paper was coordinated by XXXX. }
	
	%\blfootnote{\noindent  and G.A. Hodtani are with Department of Electrical Engineering, Ferdowsi University of Mashhad, Mashhad, Iran. (e-mail: $\rm \{f.rostami.gh,ghodtani\}@gmail.com$).}
	
	\blfootnote{\noindent This work was funded in part by Junta de Andaluc\'ia, the European Union and the European Fund for Regional Development FEDER through grants P18-RT-3175 and EMERGIA20-00297, and in part by MCIN/AEI/10.13039/501100011033 through grant PID2020-118139RB-I00.}
	\blfootnote{\noindent The authors are with the Communications and Signal Processing Lab, Telecommunication Research Institute (TELMA), Universidad de M\'alaga, M\'alaga, 29010, (Spain). F.~J.~L\'opez-Mart\'inez is also with the Dept. Signal Theory, Networking and Communications, University of Granada, 18071, Granada (Spain). (E-mails: $\rm farshad@ic.uma.es, fjlm@ugr.es$).}
	
	\blfootnote{Digital Object Identifier 10.1109/XXX.2021.XXXXXXX}
	%\IEEEpeerreviewmaketitle
	\vspace{-3mm}
	\section{Introduction}\label{introduction}
	Nowadays, the energy supply for electronic devices has become one of the most important challenges in designing future wireless communication systems, i.e., sixth-generation (6G) network \cite{david20186g}. For instance, in emerging technologies such as the Internet of Things (IoT) and its corresponding wireless applications like body wearables, peer-to-peer (P2P), device-to-device (D2D), and vehicle-to-vehicle (V2V) communications, the device nodes are mostly battery-dependent and power-constrained, and thus, they require intermittent battery replacement and recharging to maintain network connectivity, which is too expensive or even impossible in some cases. In this regard, energy harvesting (EH) from ambient energy sources has appeared as a promising approach to prolong the lifetime of energy-constrained wireless communication systems \cite{huang2012throughput,jiang2017optimal,xiong2017rate}, equipped with replacing or recharging batteries. In contrast to traditional EH technologies that mostly relied on natural energy sources and had limited ranges of applications due to environment uncertainty, recent EH technologies exploit radio frequency (RF) signals that provide reliable energy flows and guarantee the system performance. Indeed, since the RF signals are able to carry both energy and information, simultaneous wireless information and power transfer (SWIPT) has become an alternative approach to power the next generation of wireless networks. The main idea of SWIPT was first introduced in \cite{varshney2008transporting} from an information-theoretic viewpoint, where the authors proposed that nodes harvest energy from their received RF information-bearing signals. However, it is not feasible for receivers' architecture to decode signals and harvest energy at the same time due to practical limitations \cite{lu2014wireless}. Later, in order to address this issue, the authors in \cite{zhou2013wireless} proposed two practical receiver architectures with separated information decoding and energy harvesting receiver for SWIPT, namely the power splitting (PS) and the time switching (TS) architectures. In the TS protocol, the receiver switches over time between EH and information processing, whereas in the PS scheme the receiver uses a portion of the received power for EH and the remaining for information processing.
	\subsection{Related Works}
	In recent years, intense research activities have been carried out related to investigate the role of SWIPT in various wireless communication systems, especially cooperative relaying networks \cite{nasir2013relaying,pan2017outage,do2017exploiting,
		di2016simultaneous,rabie2017half,li2015outage,lee2016outage,lou2017performance,zhong2018outage,rabie2018full,nauryzbayev2018performance,makarfi2020performance,mohjazi2018performance}. 	
%multi-input multi-output systems \cite{zhang2013mimo,zhou2018energy,alageli2019optimal,wang2020massive,wen2016joint,hoang2020performance}, power allocation strategies \cite{ding2014power,huang2020power,liu2017power,yang2020resource,chen2021energy,guo2020energy,kurup2020power}, cellular networks \cite{budhiraja2021swipt,akbar2016simultaneous,huang2014enabling}, and reconfigurable intelligent surface (RIS) aided communications \cite{psomas2021swipt,yang2021optimal,yang2021beamforming,diamanti2021joint,hehao2020intelligent}. 
In \cite{nasir2013relaying}, the authors considered an amplify-and-forward (AF) relay network with Rayleigh fading channels and analyzed the key performance metrics such as ergodic capacity and outage probability under both PS and TS protocols to determine the proposed system throughput, where it was showed that the SWIPT-based relaying provides throughput improvement, communication reliability enhancement, and coverage range extension. In contrast, the authors in \cite{pan2017outage} derived the closed-form expression of the outage probability over independent Rayleigh SWIPT-relaying networks, where both AF and decode-and-forward (DF) protocols were considered. A more general SWIPT-relaying network, i.e., multiuser multirelay cooperative network, over Rayleigh fading channels was considered in \cite{do2017exploiting}, where the authors investigated the outage probability performance under DF, variable-gain AF, and fixed-gain AF protocols. Proposing two information receiving strategies, i.e., the mutual information accumulation (IA) and the energy accumulation (EA), the authors in \cite{di2016simultaneous} evaluated the achievable rate region of the SWIPT relaying network under Rayleigh fading channels. Considering log-normal fading channels in a dual-hop SWIPT relaying network, the ergodic outage probability performance for both full-duplex (FD) and half-duplex (HD) relaying mechanisms with DF and AF relaying protocols under PS and TS schemes was investigated in \cite{rabie2017half}. The outage probability performance for SWIPT relaying networks in the presence of direct link between the source and the destination under Rayleigh fading channels was analyzed in \cite{li2015outage} and  \cite{lee2016outage}. Furthermore, assuming the direct link between the source and the destination, the performance of SWIPT relaying networks in terms of the outage probability and bit error rate under Nakagami-$m$ fading channels was investigated in \cite{lou2017performance} and \cite{zhong2018outage}, respectively. On the other hand, key performance metrics for SWIPT relaying networks under generalized $\kappa$-$\mu$, $\alpha$-$\mu$, and Fisher-Snedecor $\mathcal{F}$ composite fading channels were analyzed in \cite{rabie2018full}, \cite{nauryzbayev2018performance}, and \cite{makarfi2020performance}, respectively.
	\subsection{Motivation and Contribution}
	Recent research has shown that the performance of SWIPT-based relaying networks highly depends on the statistical characteristics of channels in radio propagation environments. Therefore, accurate modeling of fading channels in SWIPT-based relaying networks is a momentous issue that should be considered. However, in all the above-mentioned literature related to the performance analysis of SWIPT relaying networks, it was ignored the potential dependence structure of the source-to-relay ($SR$) hop on the relay-to-destination ($RD$) hop, while the channel coefficients observed by the relay and the destination may be correlated in practice \cite{fan2017secure,fan2017secrecy}. In addition, from a communication-theoretic perspective, the equivalent channel observed by the destination over a SWIPT-relaying network is the product of two correlated random variables (RVs), which largely complicates the performance evaluation of such a system. On the other hand, the underlying dependence between fading channel coefficients may not be linear, and thus, the classic Pearson correlation coefficient fails to appropriately model the interdependence of fading events caused by different mechanisms, especially as to the tails of the fading distributions \cite{livieratos2014correlated}. Hence, the role of general dependence structures beyond linear correlation is gaining momentum in the wireless community. In this regard, one flexible method for incorporating both positive/negative dependence structures between RVs and describing the non-linear dependency between arbitrary RVs is copula theory, which has been recently used in the performance analysis of wireless communication systems \cite{gholizadeh2015capacity,ghadi2020copula,ghadi2020copula1,jorswieck2020copula,besser2020copula,ghadi2021role,besser2020bounds,ghadi2022capacity,ghadi2022performance}. Copula functions are mostly defined with a specific dependence parameter which indicates the measure of dependency between correlated RVs. \textcolor{blue}{In this regard, the main advantages of using copula theory rather than traditional statistical methods are: $(i)$  Providing a more general case by describing both linear and non-linear structure of dependence between two or more arbitrary RVs; $(ii)$ Generating the multivariate distributions of two or more arbitrary random variables by only knowing the marginal distributions; $(iii)$ Describing both positive and negative dependencies between two or more arbitrary RVs. Exploiting such benefits, the authors in \cite{gholizadeh2015capacity} analyzed the capacity of correlated Nakagami-$m$ fading MIMO channels by using specific copula functions under linear and positive dependence structures. The outage probability and coverage region for correlated Rayleigh fading multiple access channels were also derived in \cite{ghadi2020copula} by using the Farlie-Gumbel-Morgenstern (FGM) copula function. The authors in \cite{ghadi2020copula1} also applied the FGM copula to secure communications and derived closed-form expressions for the secrecy outage probability and average secrecy capacity under correlated Rayleigh fading channels. In \cite{jorswieck2020copula} and \cite{besser2020copula}, the authors provided the upper and lower bounds of the outage probability for the multiple access communications under correlated Rayleigh fading distributions. Using the FGM copula function, the authors in \cite{ghadi2021role} analyzed the outage probability and coverage region of multiple access communications in the presence of the non-causally known side information at the transmitters. The upper and lower bounds for secrecy performance metrics were also derived in \cite{besser2020bounds} by using the Fr\'echet-Hoeffding (FH) bounds. In \cite{ghadi2022capacity}, the authors obtained a general formulation for the product fading channels for backscatter communication systems and derived the closed-form expression of the average capacity using the FGM copula. In addition, by exploiting Clayton copula, the outage probability and average capacity for different multiple access communications under correlated Fisher-Snedecor $\mathcal{F}$ fading channels were analyzed in \cite{ghadi2022performance}. Even so, most previous works in this context only considered a specific copula function \cite{gholizadeh2015capacity,ghadi2020copula,ghadi2020copula1}, \cite{ghadi2021role}, \cite{ghadi2022performance} or considered FH bounds to analyze their proposed system models \cite{jorswieck2020copula,besser2020copula}, \cite{besser2020bounds}, %and not analyze the system performance under arbitrary fading correlation.
	 but fail to analyze the system performance under arbitrary fading correlation. References dedicated to investigate the role of correlation in the context of wireless powered communications are scarce, and mostly restricted to point-to-point links with linear dependence \cite{Tarrias2021}, or based on simulations \cite{santi2022}. However, depending on the specific set-up, correlation between links in relay-based scenarios is experimentally shown to appear in backscatter communications due to pinhole effect \cite{Griffin2007,Zhang2019}, or due to the slow variability of fading in SWIPT scenarios with bidirectional operation between the energy source and the energy-constrained device \cite{Deng2018,Tarrias2021}.} With all the aforementioned considerations, there are several unanswered practical questions over SWIPT-based relaying networks to date: $(i)$ What is the effect of fading correlation on the key performance metrics of SWIPT in cooperative relaying communications? $(ii)$ How does fading severity affect the performance of SWIPT in cooperative relaying communications?
	To the best of the authors' knowledge, there has been no previous research in analyzing SWIPT-relaying networks with arbitrarily distributed and correlated fading channels. Motivated by the aforesaid observations, we are concerned with the correlated fading issue of wireless energy harvesting and information processing over the DF cooperative relaying communications. To this end, we consider the scenario that the energy-constrained relay node harvests energy from the RF signal broadcasted by a source node and uses that harvested energy to forward the source signal to a destination node, where the $SR$ and $RD$ links are correlated RVs with arbitrary distributions. Based on the DF relaying protocol, we adopt the PS-based relaying (PSR) scheme, as proposed in \cite{nasir2013relaying}, for separate information processing and energy harvesting at the energy-constrained relay node. \textcolor{blue}{Using a mathematical approach based on that adopted in \cite{ghadi2022capacity}}, we introduce a general formulation for the cumulative distribution function (CDF) 
	%and probability density function (PDF) 
	of two correlated RVs with any arbitrary distribution, exploiting the copula theory. Then, in order to analyze the system performance, we derive the closed-form expression of the ergodic capacity and outage probability under Nakagami-$m$ fading channels, using a specific copula function. Specifically, the main
	contributions of our work are summarized as follows:
	
	\textbullet\, We provide general formulations for the CDF %and PDF 
	of the equivalent channel observed by the destination (i.e, the product of two arbitrarily distributed and correlated RVs). 
	 
	\textbullet \, In order to realize the impact of the fading correlation on the system performance, we derive the closed-form expressions of the ergodic capacity and outage probability assuming the PSR protocol under correlated Nakagami-$m$ fading, exploiting the Farlie-Gumbel-Morgenstern (FGM) copula.
	
	\textbullet\, By changing the
	dependence parameter within the defined range, our numerical and simulation results show that the system performance improves in terms of the ergodic capacity and the outage probability under the positive dependence structure, while the negative correlation has destructive effects on the system efficiency. In addition, a reduction in fading severity improves the system performance under the PSR scheme. 
	\subsection{Paper Organization}
	The rest of this paper is organized as follows. Section \ref{sec-model} describes the system model considered in our work. In section \ref{sec-snr}, the main concept of copula theory is reviewed, and then the signal-to-noise ratio (SNR) distribution is derived. Section \ref{sec-metrics} presents the main results of the considered SWIPT-based relaying network under correlated Nakagami-$m$ fading channels so that the closed-form expressions of the ergodic capacity and outage probability are determined in subsections \ref{subsec-capacity} and \ref{subsec-out}, respectively.  In section \ref{sec-results}, the efficiency of analytical results is illustrated numerically, and finally, the conclusions are drawn in section \ref{sec-conclusion}.
	\section{System Model}\label{sec-model}
	\subsection{Channel Model}
	We consider a relay network as shown in Fig. \ref{fig-model}, where a source node $S$ wants to send information to a destination node $D$ through an intermediate relay node $R$. It is assumed that there is no direct link between the source $S$ and the destination $D$ due to deep shadow fading or surrounding physical obstacles. %large distance between them.
	Such an assumption is widely adopted in research studies related to SWIPT relay communications \cite{nasir2013relaying,rabie2017half,mohjazi2018performance}. Specifically, this presumption is related to the coverage extension models where there is a long distance between the source and destination, and relays are exploited in order to maintain connectivity. This model is used in Internet-of-Thing (IoT) deployments, where RF-powered relays are employed to provide coverage expansion to avoid interference. \textcolor{blue}{As will become evident in the sequel, because of the time-division fashion on which the half-duplex relay protocol is implemented, interference is avoided at the information receiver.} For simplicity purposes, we assume that all nodes are equipped with single antennas. We also suppose that the nodes $S$ and $D$ have sufficient energy supply from other sources (e.g., a battery or a power grid), while the relay $R$ has no external power supply and only relies on the harvested signal transmitted by source $S$, thus relay $R$ is energy-constrained. The transmission consists of two phases and the HD deployment based on the DF strategy is adopted for the relay node $R$. The channel coefficients of $SR$ and $RD$ are defined by $h_{\mathrm{SR}}$ and $h_{\mathrm{RD}}$, respectively, and they are considered arbitrarily distributed and correlated RVs. Besides, we assume all channels are quasi-static fading channels, that is, the fading coefficients are fixed during the transmission of an entire codeword (i.e., $h_{\mathrm{SR}}(i)=h_{\mathrm{SR}}$ and $h_{\mathrm{RD}}(i)= h_{\mathrm{RD}}$, $\forall i= 1, ..., n$), and vary randomly from one block to another block.
\begin{figure*}
	\centering
	\hspace{-0.1cm}\subfigure[]{%
		\includegraphics[width=0.3\textwidth]{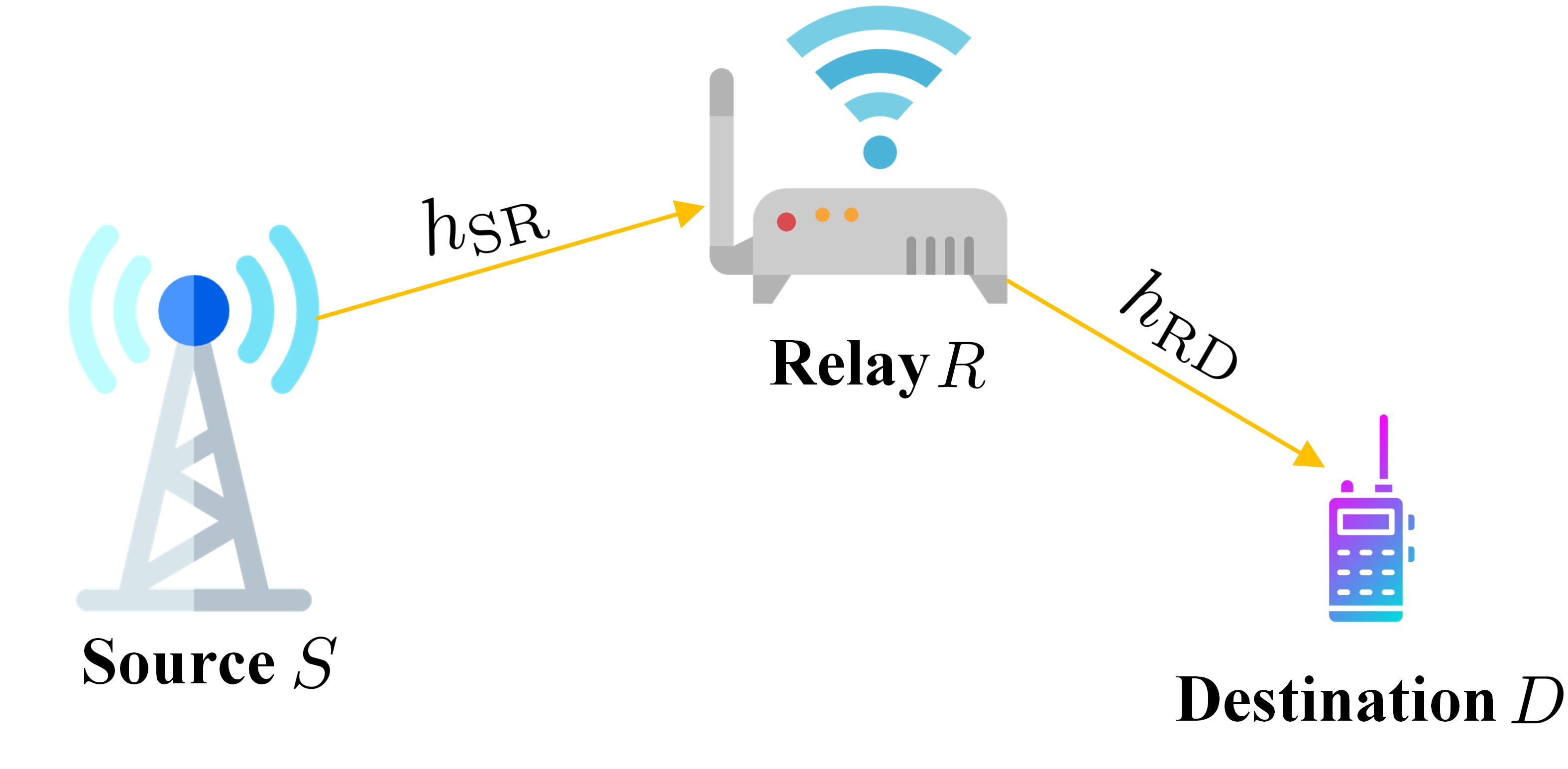}%
		\label{fig-model}%
	}\hspace{0.1cm}%or more
	\subfigure[]{%
		\includegraphics[width=0.28\textwidth]{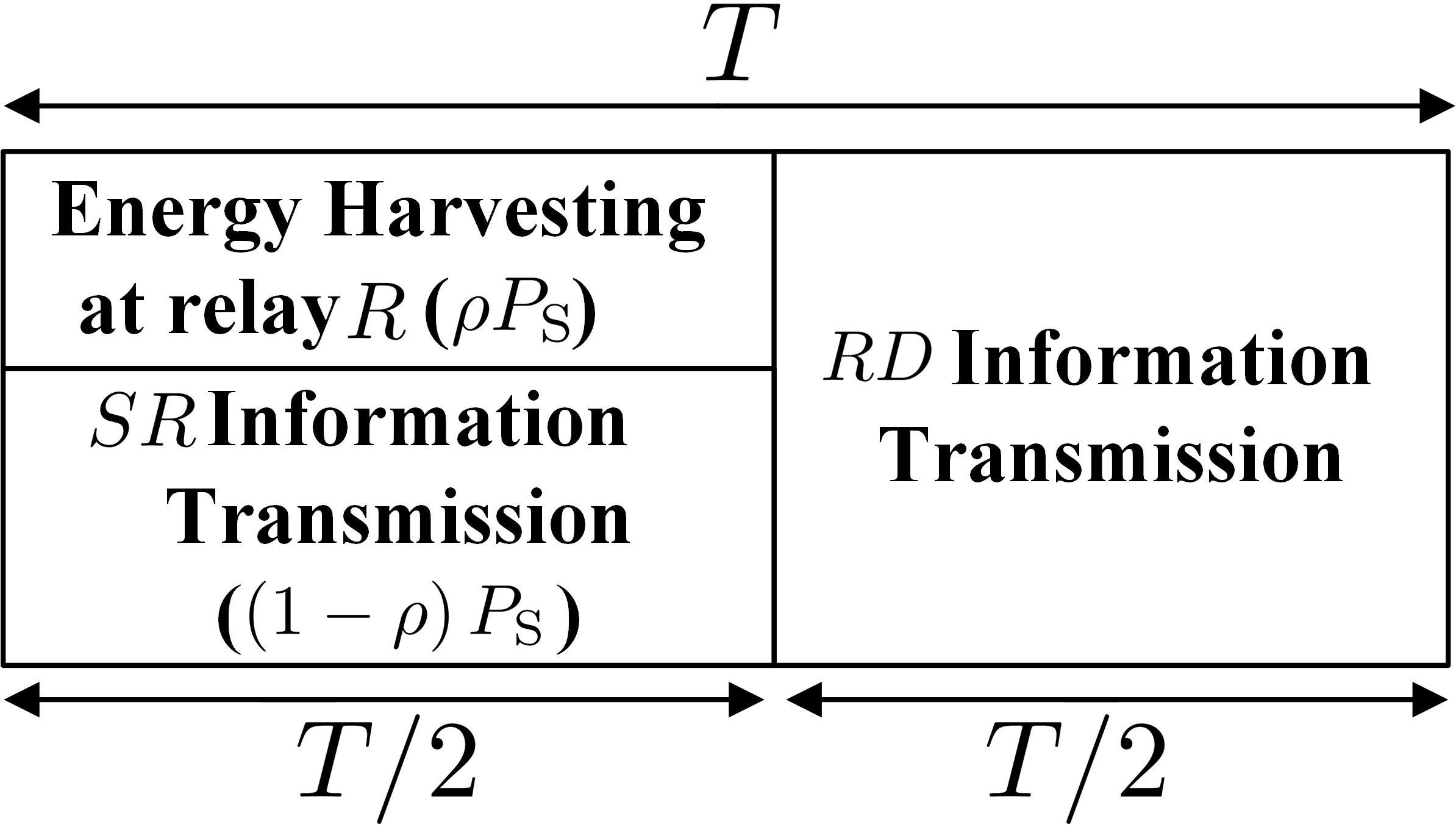}%
		\label{fig-psr1}%
	}\hspace{0.1cm}%or more
	\subfigure[]{%
		\includegraphics[width=0.4\textwidth]{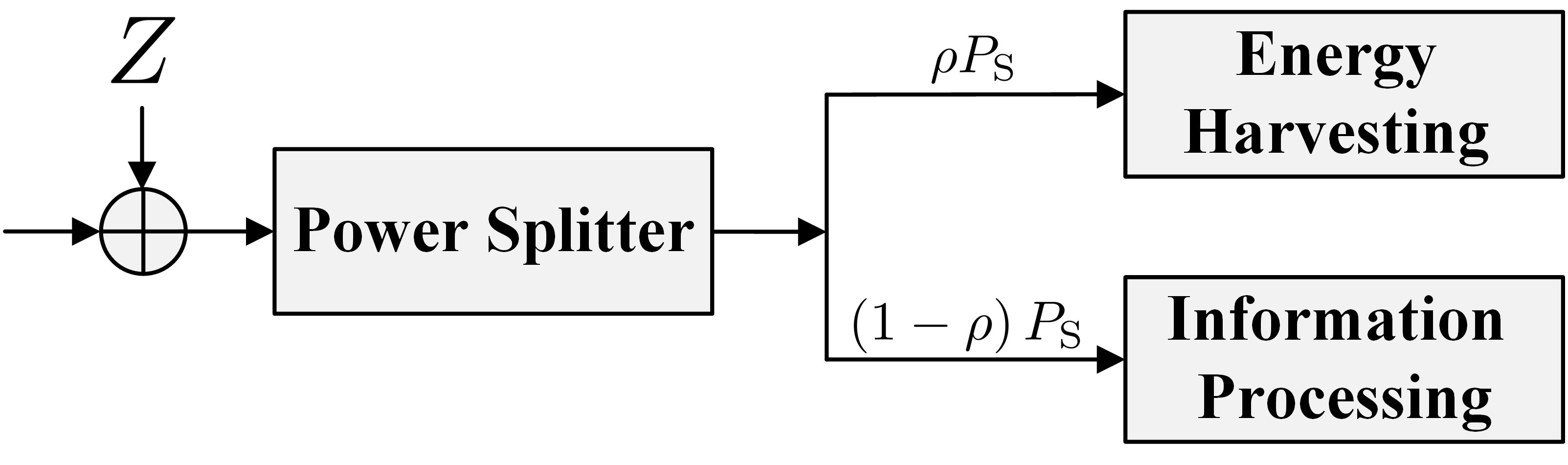}%
		\label{fig-psr2}%
	}\hspace{-0.1cm}%or more
	\caption{(a) Illustration of a SWIPT-based relaying network. (b) Illustration of the PSR protocol for energy harvesting and information processing at the relay $R$. (c)  Block diagram of the relay $R$ under the PSR protocol, where $Z$  denotes noise with power $N$.}\vspace{0cm}
	\label{fig:ab}
\end{figure*}
	\subsection{Information and Energy Transfer}
	We exploit the PSR protocol for transferring information and energy in the considered model. In this protocol, the total communication process time $T$ is split into two consecutive time slots. Let $0<\rho\leq 1$ and $P_\mathrm{S}$ denote the power-splitting factor and source power, respectively. In the first slot, the relay node $R$ uses a portion of the source power $(1-\rho)P_\mathrm{S}$ for information processing (i.e., decoding), and the remaining source power $\rho P_\mathrm{S}$ is harvested, so it can be used in the second time slot for relaying the decoded information. (see Figs. \ref{fig-psr1} and \ref{fig-psr2}). Thus, the instantaneous SNR at the relay $R$ and the destination $D$ can be defined as:
	\begin{align}
		\gamma_\mathrm{R}=\frac{(1-\rho)P_\mathrm{S}|h_\mathrm{SR}|^2}{d^{\alpha}_\mathrm{SR}N}=\hat{\gamma}_\mathrm{R}g_\mathrm{SR},
	\end{align}
	\begin{align}
		\gamma_\mathrm{D}=\frac{\kappa\rho P_\mathrm{S}|h_\mathrm{SR}|^2|h_\mathrm{RD}|^2}{d^{\alpha}_{\mathrm{SR}}d^{\alpha}_{\mathrm{RD}}N}=\hat{\gamma}_{\mathrm{D}}g_\mathrm{SR}g_\mathrm{RD},
	\end{align}
	where $d_\mathrm{SR}$ and $d_\mathrm{RD}$ are the distances of $SR$ and $RD$, respectively, and $\alpha$ is the path-loss exponent. The terms $g_\mathrm{SR}=|h_\mathrm{SR}|^2$ and \textcolor{blue}{$g_\mathrm{RD}=|h_\mathrm{RD}|^2$} define the {\textcolor{blue}{normalized}} fading power channel coefficients associated to the $SR$ and $RD$ links, respectively. Besides, $0<\kappa\leq 1$ is the EH efficiency and $N$ denotes the noise power. 
    {\textcolor{blue}{The average SNR at the destination is given by $\mathbb{E}\{\gamma_\mathrm{D}\}=\hat{\gamma}_{\mathrm{D}}\mathbb{E}\{g_\mathrm{SR}g_\mathrm{RD}\}$. Note that $\mathbb{E}\{g_\mathrm{SR}g_\mathrm{RD}\}=1$ only in the case of independence; however, in the cases with positive/negative dependence, this factor will depend on the actual relation between $g_\mathrm{SR}$ and $g_\mathrm{RD}$.}
	\section{SNR distribution}\label{sec-snr}
	In this section, we derive general analytical expressions for the CDF %and PDF 
	of SNR $\gamma_\mathrm{D}$ by exploiting the copula theory. 
	%$\gamma_{D}$ in this section. Since the SNR at destination $D$ is the product of two arbitrary dependent random variables $G_{SR}$ and $G_{RD}$, we derive analytical expressions for the CDF and PDF of $\gamma_{D}$ in this section. 
	\subsection{Copula definition and properties}
	In order to determine the distribution of $\gamma_\mathrm{D}$ in the general case, we first briefly review some fundamental definitions and properties of the two-dimensional copulas \cite{nelsen2007introduction}.
	
	\begin{definition}[Two-dimensional copula] The copula function $C(u_1,u_2)$ of a random vector $\vec{X}\left(X_1,X_2\right)$ defined on the unit hypercube $[0,1]^2$ with uniformly distributed RVs $U_j:=F_{X_j}(x_j)$ for $j\in\{1,2\}$ over $[0,1]$ is given by
		\begin{align}
			C(u_1,u_2)=\Pr(U_1\leq u_1,U_2\leq u_2).
		\end{align}
where $F_{X_j}(x_j)=\Pr(X_j\leq x_j)$ denotes the marginal CDF.
%		
%		Let $\vec{X}=(X_1,X_2)$ be a vector of two RVs with marginal CDFs $F_{X_j}(x_j)=\Pr(X_j\leq x_j)$ for $j\in\{1,2\}$, respectively. The relevant bivariate CDF is defined as:%\vspace{-1ex}
%		%
%		\begin{align}
%			F_{X_1,X_2}(x_1,x_2)=\Pr(X_1\leq x_1,X_2\leq x_2),
%		\end{align}
%		then, the copula function $C(u_1,u_2)$ of the random vector $\vec{X}$ defined on the unit hypercube $[0,1]^2$ with uniformly distributed RVs $U_j:=F_{X_j}(x_j)$ for $j\in\{1,2\}$ over $[0,1]$ is given by
%		\begin{align}
%			C(u_1,u_2)=\Pr(U_1\leq u_1,U_2\leq u_2).
%		\end{align}
	\end{definition}
	\begin{theorem}[Sklar's theorem]\label{thm-sklar}
		Let $F_{X_1,X_2}(x_1,x_2)$ be a joint CDF of RVs with marginals $F_{X_j}(x_j)$ for $j\in\{1,2\}$. Then, there exists one copula function $C(\cdot,\cdot)$ such that for all $x_j$ in the extended real line domain $\mathbb{R}$,%
		\begin{align}\label{sklar}
			F_{X_1,X_2}(x_1,x_2)=C\left(F_{X_1}(x_1),F_{X_2}(x_2)\right).
		\end{align}
	\end{theorem}
	%
%	\begin{corollary}\label{col-joint}
%		Applying the chain rule to Theorem \ref{thm-sklar}, the joint PDF $f_{X_1,X_2}(x_1,x_2)$ is given by
%		\begin{align} \nonumber
%			f_{X_1,X_2}(x_1,x_2)=f_{X_1}(x_1)f_{X_2}(x_2)c\big(F_{X_1}(x_1),F_{X_2}(x_2)\big),
%		\end{align}
%		where $c\big(F_{X_1}(x_1),F_{X_2}(x_2)\big)=\frac{\partial^2 C(F_{X_1}(x_1),F_{X_2}(x_2))}{\partial F_{X_1}(x_1)\partial F_{X_2}(x_2)}$  is the copula density function and $f_{X_j}(x_j)$ for $j\in\{1,2\}$ are the marginal PDFs, respectively.
%	\end{corollary}
	%
	\begin{definition}[Survival copula]\label{def-surv}
		Let $\vec{X}=(X_1,X_2)$ be a vector of two absolutely continuous RVs with joint CDF $F_{X_1,X_2}(x_1,x_2)$ and marginal survival functions $\overline{F}_{X_j}(x_j)=\Pr(X_j>x_j)=1-F_{X_j}(x_j)$ for $j=1,2$, the joint survival function $\overline{F}_{X_1,X_2}(x_1,x_2)$ is given by\vspace{-1ex}
		\begin{align}
			\overline{F}_{X_1,X_2}(x_1,x_2)&=\Pr(X_1>x_1,X_2>x_2)\\
			%&=\overline{F}_{X_1}(x_1)+\overline{F}_{X_2}(x_2)-1+C(1-\overline{F}(x_1),1-\overline{F}(x_2))\\
			&=\hat{C}(\overline{F}_{X_1}(x_1),\overline{F}_{X_2}(x_2)),
		\end{align}
		where $\hat{C}(u_1,u_2)=u_1+u_2-1+C(1-u_1,1-u_2)$ is the survival copula of $\vec{X}$.
	\end{definition}
%	\begin{theorem}[Fr\'echet-Hoeffding bounds]
%		\label{theo:copula bounds}
%		For any copula function $C:[0,1]^2 \mapsto [0,1]$ and any $(u_1, u_2) \in [0,1]^2$, the following bounds hold:
%		$C^{-} \prec C \prec C^{+}$; where $C_1 \prec C_2$ if $C_1(u_1, u_2) \leq C_2(u_1, u_2) \; \forall (u_1,u_2)\in[0,1]^2$, and
%		%
%		\begin{align}
%			& C^{-}(u_1, u_2) = \max(u_1 + u_2 - 1, 0) \label{lfh} \\
%			& C^{+}(u_1, u_2) = \min(u_1, u_2)\label{ufh}
%		\end{align}
%	\end{theorem} 
	%The upper and lower Fr\'echet-Hoeffding
	%bounds model extreme dependence structures. If $C=C^{-}$ the pair of RVs are said to be countermonotonic, whereas $C=C^{+}$ means that both RVs are comonotonic. 
	%
	\begin{definition}[Dependence structures]
		Consider two copula functions that verify: 
		\begin{align}
			C_1 \prec C^{\perp} \prec C_2,
		\end{align}
		where $C^{\perp}(u_1,u_2)=u_1 \cdot u_2$ is the product copula and describes the independent structure. Then, $C_1$ and $C_2$ model the negative and positive dependence structures, respectively. 
	\end{definition}
	\subsection{Arbitrary dependence}
	Since the considered fading channels are correlated, the distribution of the SNR at the destination $D$ is that of the product of two arbitrary correlated RVs. For this purpose, we exploit the following theorems to determine %the PDF and 
	the CDF of the SNR $\gamma_\mathrm{D}$.
	\begin{theorem}\label{thm-product-cdf-pdf}
		Let $\vec{X}=(X_1,X_2)$ be a vector of two absolutely continuous RVs with the corresponding copula $C$ and CDFs $F_{X_j}(x_j)$ for $j\in\{1,2\}$. Thus, the %PDF and 
		CDF of $Y=X_1X_2$ is:
%		\begin{align}%\nonumber
%			f_Y(y)=\int_{0}^{1}c\left(u,F_{X_2}\left(\tfrac{y}{F_{X_1}^{-1}(u)}\right)\right)
%			\tfrac{f_{X_2}\left(\tfrac{y}{F_{X_1}^{-1}(u)}\right)}{|F_{X_1}^{-1}(u)|} du,\label{thm-prod}
%		\end{align}
	\begin{align}\nonumber
		&F_Y(y)=F_{X_1}(0)\\
		&+\int_{0}^{1}\mathrm{sgn}\left(F_{X_1}^{-1}(u)\right)\frac{\partial}{\partial u}C\left(u,F_{X_2}\left(\frac{y}{F_{X_1}^{-1}\left(u\right)}\right)\right)du,
	\end{align}
		where $F^{-1}_{X_1}(.)$ is an inverse function of $F_{X_1}(.)$ and  $\mathrm{sgn}(.)$ defines the Sign function.
	\end{theorem}
	\begin{proof}
		The details of proof are in Appendix \ref{app-thm-product-cdf-pdf}.
	\end{proof}
	\begin{corollary}\label{col-pdf}
		The %PDF and 
		CDF of $\gamma_\mathrm{D}$  in the general dependence case of fading channels is given by
%		\begin{align}\nonumber
%			f_{\gamma_\mathrm{D}}(\gamma_\mathrm{D})=&\int_{0}^{\infty}\frac{1}{\hat{\gamma}_\mathrm{D}g_\mathrm{SR}}f_{G_\mathrm{SR}}(g_\mathrm{SR})f_{G_\mathrm{RD}}\left(\frac{\gamma_\mathrm{D}}{\hat{\gamma}_\mathrm{D}g_\mathrm{SR}}\right)\\
%			&\times 
%			c\left(F_{G_\mathrm{SR}}(g_\mathrm{SR}),F_{G_\mathrm{RD}}\left(\frac{\gamma_\mathrm{D}}{\hat{\gamma}_\mathrm{D}g_\mathrm{SR}}\right)\right)dg_\mathrm{SR},\label{pdf-gammad}
%		\end{align}
			\begin{align}\nonumber
		&F_{\gamma_{\mathrm{D}}}(\gamma_\mathrm{D})=F_{G_\mathrm{SR}}(0)+\int_{0}^{1}\mathrm{sgn}\left(g_\mathrm{SR}\right)f_{G_\mathrm{SR}}(g_\mathrm{SR})\\
		&\times\frac{\partial}{\partial F_{G_\mathrm{SR}}(g_\mathrm{SR})}C\left(F_{G_\mathrm{SR}}(g_\mathrm{SR}),F_{G_\mathrm{RD}}\left(\frac{\gamma_\mathrm{D}}{\hat{\gamma}_\mathrm{D}g_\mathrm{SR}}\right)\right)dg_\mathrm{SR}\label{def-cdf-d}.
	\end{align}
	\end{corollary}
	\begin{proof}
		Let $Y=G_\mathrm{SR}G_\mathrm{RD}$ and $u=F_{G_\mathrm{SR}}(g_\mathrm{SR})$ in Theorem \ref{thm-product-cdf-pdf}, and using the fact that $F_{\gamma_\mathrm{D}}(\gamma_\mathrm{D})=F_{Y}\big(\frac{\gamma_\mathrm{D}}{\hat{\gamma}_\mathrm{D}}\big)$%  $f_{\gamma_\mathrm{D}}(\gamma_\mathrm{D})=\frac{1}{\hat{\gamma}_\mathrm{D}}f_{Y}\big(\frac{\gamma_\mathrm{D}}{\hat{\gamma}_\mathrm{D}}\big)$, 
		, the proof is completed.
	\end{proof}
	Note that Corollary \ref{col-pdf} is valid for any arbitrary choice of fading distributions as well as copula functions. However, for exemplary purposes, we assume in the sequel that the $SR$ and $RD$ fading channel coefficients (i.e., $h_{\mathrm{SR}}$ and $h_\mathrm{RD}$) follow the Nakagami-$m$ distribution, where the parameter $m\geq 0.5$ denotes  fading severity.
	Hence, the corresponding fading power channel coefficients $g_i$ for $i\in\{SR,RD\}$ are dependent Gamma RVs, \textcolor{blue}{and} we have following marginal distributions:
	\begin{align}
		f_{G_i}(g_i)=\frac{{m_i}^{m_i}}{\Gamma(m_i){\overline{g}_i}^{m_i}}{g_i}^{m_i-1}e^{-\frac{m_i}{\overline{g}_i}g_i},\label{pdf-g}
	\end{align}
	\begin{align}
		F_{G_i}(g_i)&=1-\frac{\Gamma(m_i,\frac{m_i}{\overline{g}_i}g_i)}{\Gamma(m_i)}\\
		&=1-e^{-\frac{m_i}{\overline{g}_i}g_i}\sum_{k=0}^{m_i-1}\frac{1}{k!}\textcolor{blue}{\left(\frac{m_i}{\overline{g}_i}g_i\right)^k},\label{cdf-g}
	\end{align}
	where $\overline{g}_i=\mathbb{E}[g_i]$ are the average of corresponding fading power channel coefficients, and $m_i$ are shape parameters. 
	
	Although there are many copula functions that can be used to evaluate the structure of dependency beyond linear correlation, we exploit the FGM copula in our analysis. This choice is justified because it allows capturing both negative and positive dependencies between the RVs while offering good mathematical tractability, at the expense of a certain inability to model scenarios with strong dependencies \cite{sriboonchitta2018fgm}. \textcolor{blue}{In any case, the choice of Copula is known to depend on the specific geometry and propagation conditions for the scenario under consideration, although the ability to capture \textit{both} positive and negative dependence structures is required \cite{jorswieck2020copula}.} As will be shown in section \ref{sec-results}, the use of the FGM copula is enough for our purposes of determining the effect of negative/positive correlation between $SR$ and $RD$ links. 
	\begin{definition}\label{def-fgm}[FGM copula] The bivariate FGM copula with dependence parameter $\theta\in[-1,1]$ is defined as:
		\begin{align}\label{fgm}
			C_\mathcal{F}(u_1,u_2)=u_1u_2(1+\theta(1-u_1)(1-u_2)),
		\end{align}
		where $\theta\in[-1,0)$ and $\theta\in(0,1]$ denote the negative and positive dependence structures respectively, while $\theta=0$ \textit{always} indicates the independence structure.
	\end{definition}
	\begin{theorem}\label{thm-cdf-pdf}
		The CDF %and PDF 
		of $\gamma_\mathrm{D}$ under correlated Nakagami-$m$ fading channels using the FGM copula is given by \eqref{cdf}, %and \eqref{pdf-d}, respectively, 
		\begin{figure*}[t]
			\normalsize
			%\hrulefill
			\setcounter{equation}{13}
			\begin{align}\nonumber
				\hspace{-0.25cm}F_{\gamma_\mathrm{D}}(\gamma_\mathrm{D})=&\,1-%\frac{2{m}^{m}}{\hat{\gamma}_D^{\frac{m}{2}}\Gamma(m)}
				\sqrt{2\mathcal{B}}\Bigg(\sum_{n=0}^{m-1}a_n   \gamma_{\mathrm{D}}^{\frac{m+n}{2}} K_{n-m}\left(\zeta\sqrt{\gamma_\mathrm{D}}\right)+\theta\Bigg[\sum_{n=0}^{m-1}a_n  \gamma_\mathrm{D}^{\frac{m+n}{2}} K_{n-m}\left(\zeta\sqrt{\gamma_\mathrm{D}}\right)-\sum_{k=0}^{m-1}\sum_{n=0}^{m-1}b_{k,n} \gamma_\mathrm{D}^{\frac{k+n+m}{2}} K_{n-k-m}\left(\zeta\sqrt{2\gamma_\mathrm{D}}\right)\\
				&-\sum_{n=0}^{m-1}\sum_{l=0}^{m-1}c_{n,l} \gamma_\mathrm{D}^{\frac{l+m+n}{2}} K_{l-m+n}\left(\zeta\sqrt{2\gamma_\mathrm{D}}\right)+\sum_{k=0}^{m-1}\sum_{n=0}^{m-1}\sum_{l=0}^{m-1}d_{k,n,l}\gamma_\mathrm{D}^{\frac{k+n+l+m}{2}} K_{n+l-k-m}\left(2\zeta\sqrt{\gamma_\mathrm{D}}\right)\Bigg]\Bigg).\label{cdf}
			\end{align}
			\hrulefill
		\end{figure*}
		\begin{figure*}[t]
			\normalsize
			%\hrulefill
			\setcounter{equation}{14}
			\begin{align}\nonumber
				f_{\gamma_\mathrm{D}}(\gamma_\mathrm{D})=\,&\mathcal{B}\bigg(\gamma_\mathrm{D}^{m-1}K_{0}\left(\zeta\sqrt{\gamma_\mathrm{D}}\right)+\theta\bigg[\gamma_\mathrm{D}^{m-1}K_{0}\left(\zeta\sqrt{\gamma_\mathrm{D}}\right)-\sum_{k=0}^{m-1}q_k \gamma_\mathrm{D}^{\frac{k}{2}+m-1}
				K_{k}\left(\zeta\sqrt{2\gamma_\mathrm{D}}\right)\\
				&+\sum_{k=0}^{m-1}\sum_{n=0}^{m-1}t_{k,n}\gamma_\mathrm{D}^{\frac{k+n}{2}+m-1}K_{n-k}\left(2\zeta\sqrt{\gamma_\mathrm{D}}\right)\bigg]\bigg).\label{pdf-d}
			\end{align}
			\hrulefill
		\end{figure*}
		where $\mathcal{B}=\frac{2{m}^{2m}}{\hat{\gamma}^m_\mathrm{D}\Gamma(m)^2}$, $\zeta=\frac{2m}{\sqrt{\hat{\gamma}_\mathrm{D}}}$, and $K_v(.)$ is the modified Bessel function of the second kind and order $v$. The coefficients $a_n$, $b_{k,n}$, $c_{n,l}$, and $d_{k,n,l}$ are also respectively defined as:
		\begin{align}\nonumber
			&a_n=\frac{m^n}{\hat{\gamma}_\mathrm{D}^{\frac{n}{2}}n!}, \, b_{k,n}=\frac{m^{k+n}2^{\frac{n-k-m+2}{2}}}{\hat{\gamma}_\mathrm{D}^{\frac{k+n}{2}}k!n!}, \, c_{n,l}=\frac{m^{l+n}2^{\frac{-l+m-n}{2}}}{\hat{\gamma}_\mathrm{D}^{\frac{n+l}{2}}n!l!},\\\nonumber
			&d_{k,n,l}=\frac{2 m^{k+n+l}}{\hat{\gamma}_\mathrm{D}^{\frac{k+n+l}{2}}k!n!l!}. %\, %q_k=\frac{2^{2-\frac{k}{2}}m^k}{\hat{\gamma}_\mathrm{D}^{\frac{k}{2}}k!}, \, t_{k,n}=\frac{4m^{k+n}}{\hat{\gamma}_\mathrm{D}^{\frac{k+n}{2}}k!n!}.
		\end{align}
	\end{theorem}
	\begin{proof}
		The details of proof are in Appendix \ref{app-thm-cdf-pdf}.
	\end{proof}
	The probability density function (PDF) of $\gamma_{\mathrm{D}}$ was also obtained in \cite[Thm. 4]{ghadi2022capacity} as \eqref{pdf-d}, where the coefficients $q_k$ and $t_{k,n}$ are given by $q_k=\frac{2^{2-\frac{k}{2}}m^k}{\hat{\gamma}_\mathrm{D}^{\frac{k}{2}}k!}$ and  $t_{k,n}=\frac{4m^{k+n}}{\hat{\gamma}_\mathrm{D}^{\frac{k+n}{2}}k!n!}$, respectively. It should be noted that the closed-form expressions of the CDF and PDF provided in \eqref{cdf} and \eqref{pdf-d} are valid for integer values of $m$, while the integral-form expressions can be used for arbitrary positive real values of $m$. In addition, the PDF and the CDF of $\gamma_\mathrm{R}$ are also given by
	\begin{align}
		f_{\gamma_\mathrm{R}}(\gamma_\mathrm{R})=\frac{{m}^{m}}{\Gamma(m){\hat{\gamma}_\mathrm{R}}^{m}}{\gamma_\mathrm{R}}^{m-1}\mathrm{e}^{-\frac{m}{\hat{\gamma}_\mathrm{R}}\gamma_\mathrm{R}},\label{pdf-r}
	\end{align}
	\begin{align}
		F_{\gamma_\mathrm{R}}(\gamma_\mathrm{R})=1-\mathrm{e}^{-\frac{m}{\hat{\gamma}_\mathrm{R}}\gamma_\mathrm{R}}\sum_{k=0}^{m-1}\frac{1}{k!}\textcolor{blue}{\left(\frac{m}{\hat{\gamma}_\mathrm{R}}\gamma_\mathrm{R}\right)^k}.\label{cdf-r}
	\end{align}

	\begin{corollary}\label{col-expected}
		\textcolor{blue}{The average SNR $\overline\gamma_{\mathrm{D}}=\hat\gamma_{\mathrm{D}}\mathbb{E}\{g_\mathrm{SR}g_\mathrm{RD}\}$ is given by}
%		\begin{align}\nonumber
%			f_{\gamma_\mathrm{D}}(\gamma_\mathrm{D})=&\int_{0}^{\infty}\frac{1}{\hat{\gamma}_\mathrm{D}g_\mathrm{SR}}f_{G_\mathrm{SR}}(g_\mathrm{SR})f_{G_\mathrm{RD}}\left(\frac{\gamma_\mathrm{D}}{\hat{\gamma}_\mathrm{D}g_\mathrm{SR}}\right)\\
%			&\times 
%			c\left(F_{G_\mathrm{SR}}(g_\mathrm{SR}),F_{G_\mathrm{RD}}\left(\frac{\gamma_\mathrm{D}}{\hat{\gamma}_\mathrm{D}g_\mathrm{SR}}\right)\right)dg_\mathrm{SR},\label{pdf-gammad}
%		\end{align}
    
	\begin{equation}
	\textcolor{blue}{\overline\gamma_{\mathrm{D}}=\hat\gamma_{\mathrm{D}}\left(1+\theta\left(1-\mathbb{f}(m)\right)\right)\label{def-exp-d},}
	\end{equation}
    {\textcolor{blue}{where}} 
    \begin{equation}
        \textcolor{blue}{\mathbb{f}(m)=\sum_{k=0}^{m-1}\tfrac{2^{-\left(m+k-1\right)}}{m {\mathrm{B}}(m,k+1)}-\sum_{k=0}^{m-1}\sum_{n=0}^{m-1}\tfrac{2^{-\left(2m+k+n\right)}}{m^2 {\mathrm{B}}(m,k+1){\mathrm{B}}(m,n+1)},}
    \end{equation}
    \textcolor{blue}{and ${\mathrm{B}}(\cdot,\cdot)$ is the Beta function.}
	\end{corollary}
	\begin{proof} \textcolor{blue}{Using \eqref{pdf-d} and the definition of expectation, and after some manipulations, \eqref{def-exp-d} is obtained.}
        \end{proof}

\textcolor{blue}{Corollary \ref{col-expected} shows that the average SNR at the destination is increased (decreased) due to positive (negative) dependence by a scaling factor that depends on the fading severity, compared to the case of independence. %This is represented in Fig. X, on which the factor $\overline{\gamma}_D/\hat{\gamma}_D$ is represented for different values of $m$ and $\theta$. 
This indicates that the role of positive/negative dependence has an impact on the average SNR, which in turn will impact the system performance, as discussed next. In the sequel,} we exemplify how the key performance metrics of interest, i.e., the ergodic capacity and outage probability, can be characterized in the closed-form expressions for the specific case of using the Nakagami-$m$ fading and the FGM copula. 
	\section{Performance Analysis: Ergodic Capacity and Outage Probability}\label{sec-metrics}
	In this section, we derive analytical expressions for the ergodic capacity and the outage probability for the considered system model under dependent Nakagami-$m$ fading channels.
	\subsection{Ergodic Capacity}\label{subsec-capacity}
	In the considered dual-hop relay network, the instantaneous capacity is defined as \cite{laneman2004cooperative}:
	\begin{align}
		\mathcal{C}=\min(\mathcal{C}_\mathrm{SR},\mathcal{C}_\mathrm{RD}),\label{c-def}
	\end{align}
	where $\mathcal{C}_\mathrm{SR}$ and $\mathcal{C}_\mathrm{RD}$ are the instantaneous capacity of the $SR$ and $RD$ links, respectively, which can defined as follows:
	\begin{align}
		\mathcal{C}_\mathrm{SR}=\frac{1}{2}\log_2\left(1+\gamma_\mathrm{R}\right),\label{c-sr}
	\end{align}
	\begin{align}
		\mathcal{C}_\mathrm{RD}=\frac{1}{2}\log_2\left(1+\gamma_\mathrm{D}\right).\label{c-rd1}
	\end{align}
	
	\begin{theorem}
		The ergodic capacity of the $SR$ link for the considered system model under Nakagami-$m$ fading channel is given by
		\begin{align}
			\overline{\mathcal{C}}_\mathrm{SR}=\frac{1}{2\Gamma(m)\ln 2}G_{3,2}^{1,3}\left(\begin{array}{c}
				\frac{\hat{\gamma}_\mathrm{R}}{m}\end{array}
			\Bigg\vert\begin{array}{c}
				\left(1-\frac{m}{\hat{\gamma}_\mathrm{R}},1,1\right)\\
				(1,0)\\
			\end{array}\right).\label{c-srf}
		\end{align}
	\end{theorem}
	\begin{proof}
		The ergodic capacity given in \eqref{c-sr} can be further mathematically expressed as:
		\begin{align}
			\overline{\mathcal{C}}_\mathrm{SR}&=\frac{1}{2\ln 2}\int_{0}^{\infty}\ln\left(1+\gamma_\mathrm{R}\right)f_{\gamma_\mathrm{R}}(\gamma_\mathrm{R})d\gamma_\mathrm{R},\label{def-cr}
		\end{align}
		where $f_{\gamma_\mathrm{R}}(\gamma_\mathrm{R})$ is given by \eqref{pdf-r}. Next, by re-expressing the logarithm function in terms of the Meijer's G-function \cite[Eq. (11)]{adamchik1990algorithm}, i.e.,
		\begin{align}
			\ln(1+x)=
			G_{2,2}^{1,2}\left(\begin{array}{c}
				x\end{array}
			\Bigg\vert\begin{array}{c}
				(1,1)\\
				(1,0)\\
			\end{array}\right),\label{log}
		\end{align}
		substituting \eqref{pdf-r} and \eqref{log} in \eqref{def-cr}, $\overline{\mathcal{C}}_\mathrm{SR}$ can be re-written as:
		\begin{align}\nonumber
			\overline{\mathcal{C}}_\mathrm{SR}&=\frac{m^m}{2\hat{\gamma}_\mathrm{R}^m\Gamma(m)\ln 2}\\
			&\times\underset{\mathcal{I}}{\underbrace{\int_{0}^{\infty}\gamma^{m-1}_\mathrm{R}\mathrm{e}^{-\frac{m}{\overline{\gamma}_\mathrm{R}}\gamma_\mathrm{R}}G_{2,2}^{1,2}\left(\begin{array}{c}
						\gamma_\mathrm{R}\end{array}
					\Bigg\vert\begin{array}{c}
						(1,1)\\
						(1,0)\\
					\end{array}\right)d\gamma_\mathrm{R}}}.\label{def-cr2}
		\end{align}
		With the help of \cite[Eq. (2.24.3.1)]{prudnikov1990more}, $\mathcal{I}$ can be computed as:
		\begin{align}
			\mathcal{I}=\frac{\hat{\gamma}_\mathrm{R}^m}{m^m}G_{3,2}^{1,3}\left(\begin{array}{c}
				\frac{\hat{\gamma}_\mathrm{R}}{m}\end{array}
			\Bigg\vert\begin{array}{c}
				\left(1-\frac{m}{\hat{\gamma}_\mathrm{R}},1,1\right)\\
				(1,0)\\
			\end{array}\right).\label{I}
		\end{align}
		Now, by inserting \eqref{I} into \eqref{def-cr2} the proof is completed.
	\end{proof}
	\begin{theorem}
		The ergodic capacity of the $RD$ link for the considered system model under Nakagami-$m$ fading channel is given by \eqref{c-rd}, where $\mathcal{D}$, $w_k$, and $z_{k,n}$ are respectively defined as:
		\begin{align}\nonumber
			\mathcal{D}=\frac{2^{2m-2}\mathcal{B}}{\pi\zeta^{2m}\ln2},\quad w_k=\frac{2^{2-m}m^k}{\hat{\gamma}_\mathrm{D}^\frac{k}{2}\zeta^{k}k!},\quad z_{k,n}\frac{2^{2-2m}m^{k+n}}{\hat{\gamma}_\mathrm{D}^{\frac{k+n}{2}}\zeta^{n+k}k!n!}.	
		\end{align}
		\begin{figure*}[!t]
			\normalsize
			%\hrulefill
			\setcounter{equation}{25}
			\begin{align}\nonumber
				&\overline{\mathcal{C}}_\mathrm{RD}=\,\Bigg(G_{4,2}^{1,4}\left(\begin{array}{c}
					\frac{4}{\zeta^2}\end{array}
				\Bigg\vert\begin{array}{c}
					\left(1-m,1-m,1,1\right)\\
					\left(1,0\right)\\
				\end{array}\right)+\theta\Bigg[G_{4,2}^{1,4}\left(\begin{array}{c}
					\frac{4}{\zeta^2}\end{array}
				\Bigg\vert\begin{array}{c}
					\left(1-m,1-m,1,1\right)\\
					\left(1,0\right)\\
				\end{array}\right)\\
				&-\sum_{k=0}^{m-1}w_kG_{4,2}^{1,4}\left(\begin{array}{c}
					\frac{2}{\zeta^2}\end{array}
				\Bigg\vert\begin{array}{c}
					\left(1-(m+k),1-m,1,1\right)\\
					\left(1,0\right)\\
				\end{array}\right)+\sum_{k=0}^{m-1}\sum_{n=0}^{m-1}z_{k,n}G_{4,2}^{1,4}\left(\begin{array}{c}
					\frac{1}{\zeta^2}\end{array}
				\Bigg\vert\begin{array}{c}
					\left(1-(m+n),1-(m+k),1,1\right)\\
					(1,0)\\
				\end{array}\right)\Bigg]\Bigg).\label{c-rd}
			\end{align}
			\hrulefill
		\end{figure*}
	\end{theorem}
	\begin{proof}
		The ergodic capacity given in \eqref{c-rd1} can be further mathematically defined as:
		\begin{align}
			\overline{\mathcal{C}}_\mathrm{RD}&=\frac{1}{2\ln 2}\int_{0}^{\infty}\ln\left(1+\gamma_\mathrm{D}\right)f_{\gamma_\mathrm{D}}(\gamma_\mathrm{D})d\gamma_\mathrm{D},\label{def-cd}
		\end{align}
		where $f_{\gamma_\mathrm{D}}(\gamma_\mathrm{D})$ is determined by Theorem \ref{thm-cdf-pdf} as \eqref{pdf-d}. Thus, by plugging \eqref{log} and \eqref{pdf-d} into \eqref{def-cd}, $\overline{\mathcal{C}}_\mathrm{RD}$ can be re-expressed as:
		\begin{align}\nonumber
			\overline{\mathcal{C}}_\mathrm{RD}=\,&\frac{\mathcal{B}}{2\ln 2}\Bigg(\mathcal{J}_1+\theta\Bigg[\mathcal{J}_1-\sum_{k=0}^{m-1}\frac{2^{2-\frac{k}{2}}m^k}{\hat{\gamma}_\mathrm{D}^{\frac{k}{2}}k!}\mathcal{J}_2\\
			&+\sum_{k=0}^{m-1}\sum_{n=0}^{m-1}\frac{4m^{k+n}}{\hat{\gamma}_\mathrm{D}^{\frac{k+n}{2}}k!n!}\mathcal{J}_3\Bigg]\Bigg),\label{c-rd2}
		\end{align}
		where
		\begin{align}
			\mathcal{J}_1=\int_{0}^{\infty}\gamma_\mathrm{D}^{m-1}K_0\left(\zeta\sqrt{\gamma_\mathrm{D}}\right)G_{2,2}^{1,2}\left(\begin{array}{c}
				\gamma_\mathrm{D}\end{array}
			\Bigg\vert\begin{array}{c}
				(1,1)\\
				(1,0)\\
			\end{array}\right)d\gamma_\mathrm{D},
		\end{align}
		\small\begin{align}
			\mathcal{J}_2=\int_{0}^{\infty}\gamma_\mathrm{D}^{\frac{k}{2}+m-1}
			K_{k}\left(\zeta\sqrt{2\gamma_\mathrm{D}}\right)G_{2,2}^{1,2}\left(\begin{array}{c}
				\gamma_\mathrm{D}\end{array}
			\Bigg\vert\begin{array}{c}
				(1,1)\\
				(1,0)\\
			\end{array}\right)d\gamma_\mathrm{D},
		\end{align}
		\small\begin{align}
			\mathcal{J}_3=\int_{0}^{\infty}\gamma_\mathrm{D}^{\frac{k+n}{2}+m-1}K_{n-k}\left(2\zeta\sqrt{\gamma_\mathrm{D}}\right)G_{2,2}^{1,2}\left(\begin{array}{c}
				\gamma_\mathrm{D}\end{array}
			\Bigg\vert\begin{array}{c}
				(1,1)\\
				(1,0)\\
			\end{array}\right)d\gamma_\mathrm{D}.
		\end{align}
		With the help of \cite[(2.24.4.3)]{prudnikov1990more}, the integrals $\mathcal{J}_1$, $\mathcal{J}_2$, and $\mathcal{J}_3$ can be respectively computed as follows:
		\begin{align}
			\mathcal{J}_1=\frac{2^{2m}}{2\pi\zeta^{2m}}G_{4,2}^{1,4}\left(\begin{array}{c}
				\frac{4}{\zeta^2}\end{array}
			\Bigg\vert\begin{array}{c}
				\left(1-m,1-m,1,1\right)\\
				\left(1,0\right)\\
			\end{array}\right)\label{j1},
		\end{align}
		\begin{align}
			\mathcal{J}_2=\frac{2^{m+\frac{k}{2}}}{2\pi\zeta^{2m+k}}G_{4,2}^{1,4}\left(\begin{array}{c}
				\frac{2}{\zeta^2}\end{array}
			\Bigg\vert\begin{array}{c}
				\left(1-(m+k),1-m,1,1\right)\\
				\left(1,0\right)\\
			\end{array}\right)\label{j2},
		\end{align}
		\begin{align}\nonumber
			&\mathcal{J}_3=\frac{1}{2\pi\zeta^{2m+n+k}}\\
			&\times G_{4,2}^{1,4}\left(\begin{array}{c}
				\frac{1}{\zeta^2}\end{array}
			\Bigg\vert\begin{array}{c}
				\left(1-(m+n),1-(m+k),1,1\right)\\
				(1,0)\\
			\end{array}\right)\label{j3}.
		\end{align}
		Now, by inserting \eqref{j1}, \eqref{j2}, and \eqref{j3} into \eqref{c-rd2}, the proof is completed. 
	\end{proof}
	\subsection{Outage Probability}\label{subsec-out}
	The outage probability is defined as the probability that the received SNR is less than a certain threshold
	$\gamma_\mathrm{t}$. Thus, we define the outage probability for the given dual-hop relay network as follows:
	\begin{align}
		P_\mathrm{o}=\Pr\left(\min(\gamma_\mathrm{R},\gamma_\mathrm{D})\leq\gamma_\mathrm{t}\right).\label{def-out}
	\end{align}
	\begin{theorem}
		The outage probability for the considered dual-hop SWIPT relay network over dependent Nakagami-$m$ fading channels is given by \eqref{out}.
		\begin{figure*}[t]
			\normalsize
			%\hrulefill
			\setcounter{equation}{35}
			\small\begin{align}\nonumber
				&\hspace{-0.3cm}P_\mathrm{o}=%\frac{2{m}^{m}}{\hat{\gamma}_D^{\frac{m}{2}}\Gamma(m)}
				1-\Bigg[\frac{\Gamma(m,\frac{m}{\hat{\gamma}_\mathrm{R}}\gamma_\mathrm{t})\sqrt{2\mathcal{B}}}{\Gamma(m)}\Bigg(\sum_{n=0}^{m-1}a_n   \gamma_\mathrm{t}^{\frac{m+n}{2}} K_{n-m}\left(\zeta\sqrt{\gamma_\mathrm{t}}\right)+\theta\Bigg[\sum_{n=0}^{m-1}a_n  \gamma_\mathrm{t}^{\frac{m+n}{2}} K_{n-m}\left(\zeta\sqrt{\gamma_\mathrm{t}}\right)-\sum_{k=0}^{m-1}\sum_{n=0}^{m-1}b_{k,n} \gamma_\mathrm{t}^{\frac{k+n+m}{2}} K_{n-k-m}\left(\zeta\sqrt{2\gamma_\mathrm{t}}\right)\\\nonumber
				&-\sum_{n=0}^{m-1}\sum_{l=0}^{m-1}c_{n,l} \gamma_\mathrm{t}^{\frac{l+m+n}{2}} K_{l-m+n}\left(\zeta\sqrt{2\gamma_\mathrm{t}}\right)+\sum_{k=0}^{m-1}\sum_{n=0}^{m-1}\sum_{l=0}^{m-1}d_{k,n,l}\gamma_\mathrm{t}^{\frac{k+n+l+m}{2}} K_{n+l-k-m}\left(2\zeta\sqrt{\gamma_\mathrm{t}}\right)\Bigg]\Bigg)\\\nonumber
				&\times\Bigg(1+\theta\Bigg(1-%\frac{2{m}^{m}}{\hat{\gamma}_D^{\frac{m}{2}}\Gamma(m)}
				\sqrt{2\mathcal{B}}\Bigg(\sum_{n=0}^{m-1}a_n   \gamma_\mathrm{t}^{\frac{m+n}{2}} K_{n-m}\left(\zeta\sqrt{\gamma_\mathrm{t}}\right)+\theta\Bigg[\sum_{n=0}^{m-1}a_n  \gamma_\mathrm{t}^{\frac{m+n}{2}} K_{n-m}\left(\zeta\sqrt{\gamma_\mathrm{t}}\right)-\sum_{k=0}^{m-1}\sum_{n=0}^{m-1}b_{k,n} \gamma_\mathrm{t}^{\frac{k+n+m}{2}} K_{n-k-m}\left(\zeta\sqrt{2\gamma_\mathrm{t}}\right)\\
				&-\sum_{n=0}^{m-1}\sum_{l=0}^{m-1}c_{n,l} \gamma_\mathrm{t}^{\frac{l+m+n}{2}} K_{l-m+n}\left(\zeta\sqrt{2\gamma_\mathrm{t}}\right)+\sum_{k=0}^{m-1}\sum_{n=0}^{m-1}\sum_{l=0}^{m-1}d_{k,n,l}\gamma_\mathrm{t}^{\frac{k+n+l+m}{2}} K_{n+l-k-m}\left(2\zeta\sqrt{\gamma_\mathrm{t}}\right)\Bigg]\Bigg)\Bigg)\Bigg(1-\frac{\Gamma(m,\frac{m}{\hat{\gamma}_\mathrm{R}}\gamma_\mathrm{t})}{\Gamma(m)}\Bigg)\Bigg],\label{out}
			\end{align}
			\hrulefill
		\end{figure*}
	\end{theorem}
	\begin{proof}
		The outage probability given in \eqref{def-out} can be expressed in terms of the survival copula as follows:
		\begin{align}
			P_\mathrm{o}&=1-\Pr\left(\gamma_\mathrm{R}>\gamma_\mathrm{t},\gamma_\mathrm{D}>\gamma_\mathrm{t}\right)\\
			&=1-\hat{\mathcal{C}}\left(\overline{F}_{\gamma_\mathrm{R}}(\gamma_\mathrm{t}),\overline{F}_{\gamma_\mathrm{D}}(\gamma_\mathrm{t})\right),\label{out-def}
		\end{align}
		where $\overline{F}_{\gamma_\mathrm{R}}(\gamma_\mathrm{t})=1-F_{\gamma_\mathrm{R}}(\gamma_\mathrm{t})$ and $\overline{F}_{\gamma_\mathrm{D}}(\gamma_\mathrm{t})=1-F_{\gamma_\mathrm{D}}(\gamma_\mathrm{t})$ are the survival functions of $\gamma_\mathrm{R}$ and $\gamma_\mathrm{D}$, respectively. Now, using the fact that the FGM survival copula is same as the FGM copula, i.e., $\hat{C}_{\mathcal{F}}(u_1, u_2)=C_{\mathcal{F}}(u_1, u_2)$, inserting \eqref{cdf} and \eqref{cdf-r} into \eqref{out-def}, and doing some simplifications, the proof is completed.
	\end{proof}
\subsection{Asymptotic Analysis}
\textcolor{blue}{Although the analytical expressions provided in \eqref{c-srf}, \eqref{c-rd}, and \eqref{out} fully capture the system performance in terms of the ergodic capacity and the outage probability, we are interested in examining the asymptotic behavior of these metrics in the high SNR regime which can provide deeper insights into the system performance. For this purpose, we allow $\hat\gamma_\mathrm{R}$ to approach infinity; this can be justified when the relay is closer to the source, which is often the case for energy-constrained devices harvesting power from an external source.}

\textcolor{blue}{In the following theorem, we derive the asymptotic expression of the ergodic capacity for the considered model.
\begin{theorem}
	In the high SNR regime, the asymptotic ergodic capacity of the $SR$ link for the considered system model under Nakagami-$m$ fading channel is given by
	\begin{align}
	\overline{\mathcal{C}}_{\mathrm{SR}}^{\infty}=\frac{1}{2\Gamma(m)\ln 2}\left(-\ln\left(\frac{m}{\hat{\gamma}_\mathrm{R}}+\psi(m)\right)\right),\label{csr-asy}
	\end{align}
where $\psi(.)$ is the digamma function.  
\begin{proof}
Assuming the high SNR regime (i.e., $\gamma_{\mathrm{R}}\gg 1$), \eqref{def-cr} can be written as:
		\begin{align}
	\overline{\mathcal{C}}_\mathrm{SR}^\infty&=\frac{1}{2\ln 2}\int_{0}^{\infty}\ln\left(\gamma_\mathrm{R}\right)f_{\gamma_\mathrm{R}}(\gamma_\mathrm{R})d\gamma_\mathrm{R}\label{c-asy-def}\\
	&\overset{(a)}{=}\frac{m^m}{2\hat{\gamma}_\mathrm{R}^m\ln 2}\int_0^\infty \gamma_{\mathrm{R}}^{m-1}\mathrm{e}^{-\frac{m}{\hat{\gamma}_\mathrm{R}}}\ln\left(\gamma_{\mathrm{R}}\right)d\gamma_{\mathrm{R}},\label{c-asy-def2}
\end{align}
where $(a)$ is achieved by inserting \eqref{pdf-r} into \eqref{c-asy-def}. By solving the integral in \eqref{c-asy-def2} with the help of integral format provided in \cite{integral1}, the proof is completed.
\end{proof}
\end{theorem}}

\textcolor{blue}{The asymptotic outage probability for the proposed model is given by the following theorem.
\begin{theorem}
In the high SNR regime, the asymptotic outage probability for the considered dual-hop SWIPT relay network over dependent Nakagami-$m$ fading channels is given by \eqref{out-asy}
\begin{align}
		P_\mathrm{o}^{\infty}\approx 1-\overline{F}^{\infty}_{\gamma_\mathrm{R}}(\gamma_\mathrm{t})\overline{F}_{\gamma_\mathrm{D}}(\gamma_\mathrm{t})\big(1+\theta {F}^{\infty}_{\gamma_\mathrm{R}}(\gamma_\mathrm{t}) {F}_{\gamma_\mathrm{D}}(\gamma_\mathrm{t})\big)\label{out-asy},
\end{align}
where 
\begin{align}\label{approx-cdf}
F^{\infty}_{\gamma_\mathrm{R}}(\gamma_\mathrm{R})\approx\frac{m^m\gamma_\mathrm{R}^m}{\hat{\gamma}_\mathrm{R}^m\Gamma(m+1)},
\end{align}
and $\overline{F}^{\infty}_{\gamma_\mathrm{R}}(\gamma_\mathrm{R})=1- F^{\infty}_{\gamma_\mathrm{R}}(\gamma_\mathrm{R})$
\begin{proof}
The CDF of $\gamma_\mathrm{R}$ at the high SNR regime can be approximated by \eqref{approx-cdf} as in \cite{Wang2003}. Now, by exploiting the FGM copula definition from \eqref{fgm}, the outage probability in \eqref{out-def} can be re-written as:
	\begin{align}
		P_\mathrm{o}=1-\overline{F}_{\gamma_\mathrm{R}}(\gamma_\mathrm{t})\overline{F}_{\gamma_\mathrm{D}}(\gamma_\mathrm{t})\big(1+\theta {F}_{\gamma_\mathrm{R}}(\gamma_\mathrm{t}) {F}_{\gamma_\mathrm{D}}(\gamma_\mathrm{t})\big)\label{asy-out1}.
	\end{align}
Then, by inserting the CDF of $\gamma_{\mathrm{R}}$ from \eqref{cdf} and the approximate CDF of $\gamma_{\mathrm{D}}$ from \eqref{approx-cdf} into \eqref{asy-out1}, the asymptotic outage probability can be obtained as \eqref{out-asy} and the proof is completed. 
	\end{proof}
\end{theorem}
	\begin{remark}
In view of \eqref{csr-asy}, similar to the exact ergodic capacity derived for the SR link in \eqref{c-srf}, we can see that the correlation does not affect the ergodic capacity of the SR link at the high SNR regime for a fixed transmit power. However, regarding the impact of correlation on the ergodic capacity of the RD link in \eqref{c-rd} and the ergodic capacity definition in \eqref{c-rd}, a positive (negative) dependence increases (decreases) the ergodic capacity of the proposed system model in the high SNR regime, by virtue of Corollary 2.
		\end{remark}
\begin{remark}
Given \eqref{out-asy}, similar to the exact outage probability obtained in \eqref{out}, we state that the positive (negative) dependence provides a lower (higher) outage probability in the high SNR regime with a fixed transmit power. In addition, the asymptotic outage probability in \eqref{out-asy} highly depends on the fading parameter $m$.
%However, by comparing the asymptotic outage probability in \eqref{asy-out} with the exact outage probability in \eqref{out} and also considering the dependence parameter interval (i.e, $\theta\in [-1,1]$), the effect of fading correlation is more noticeable at the high SNR regime.
\end{remark}}

	\begin{figure}[!h]\vspace{0ex}
		\centering
		\includegraphics[width=1\columnwidth]{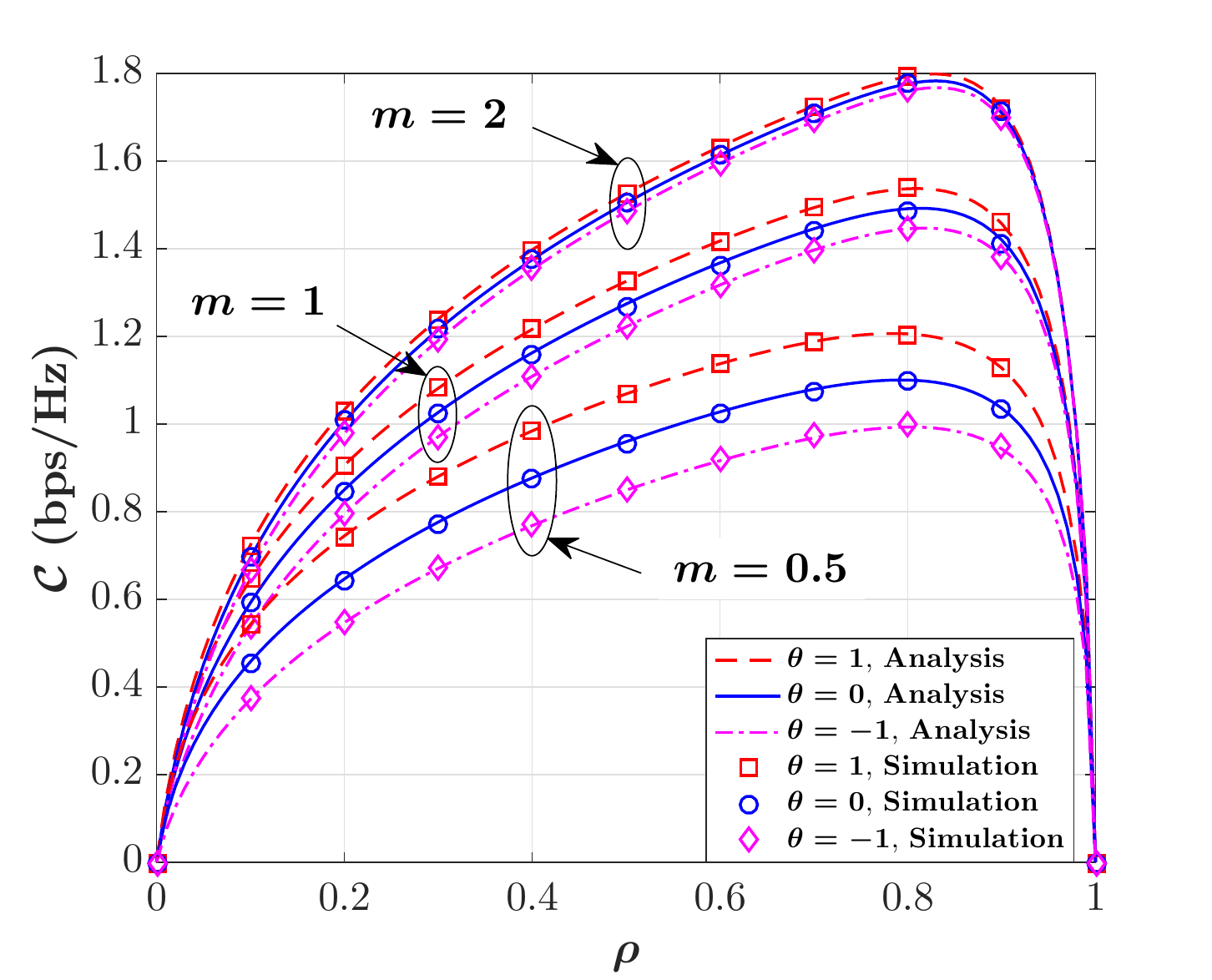} %\textwidth
		\caption{Ergodic capacity versus PSR factor $\rho$ under dependent/independent Nakagami-$m$ fading channels when  $\kappa=0.7$, $P_\mathrm{S}=10$W, $N=10^{-2}$W, $d_\mathrm{SR}=d_\mathrm{RD}=2$m, and $\alpha=2.5$.} %\vspace{-0.6cm}
		\label{fig-crhom}
	\end{figure}%\vspace{0ex}
	\begin{figure}[!h]\vspace{0ex}
		\centering
		\includegraphics[width=1\columnwidth]{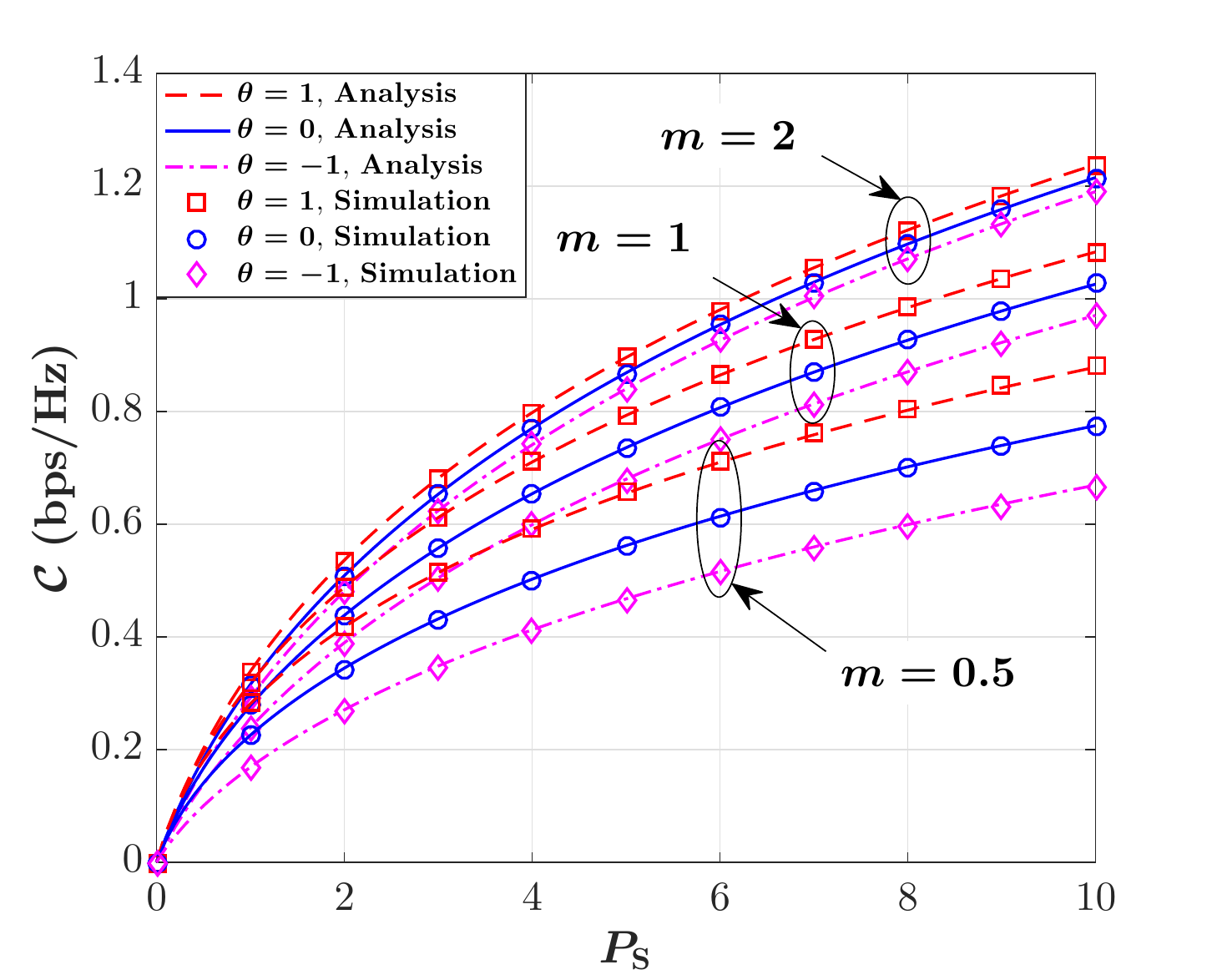} %\textwidth
		\caption{Ergodic capacity versus source power $P_\mathrm{S}$ under dependent/independent Nakagami-$m$ fading channels when $\kappa=0.7$, $\rho=0.3$, $N=10^{-2}$W, $d_\mathrm{SR}=d_\mathrm{RD}=2$m, and $\alpha=2.5$.} %\vspace{-0.6cm}
		\label{fig-cpsm}
	\end{figure}%\vspace{0ex}
	\begin{figure}[!h]\vspace{0ex}
		\centering
		\includegraphics[width=1\columnwidth]{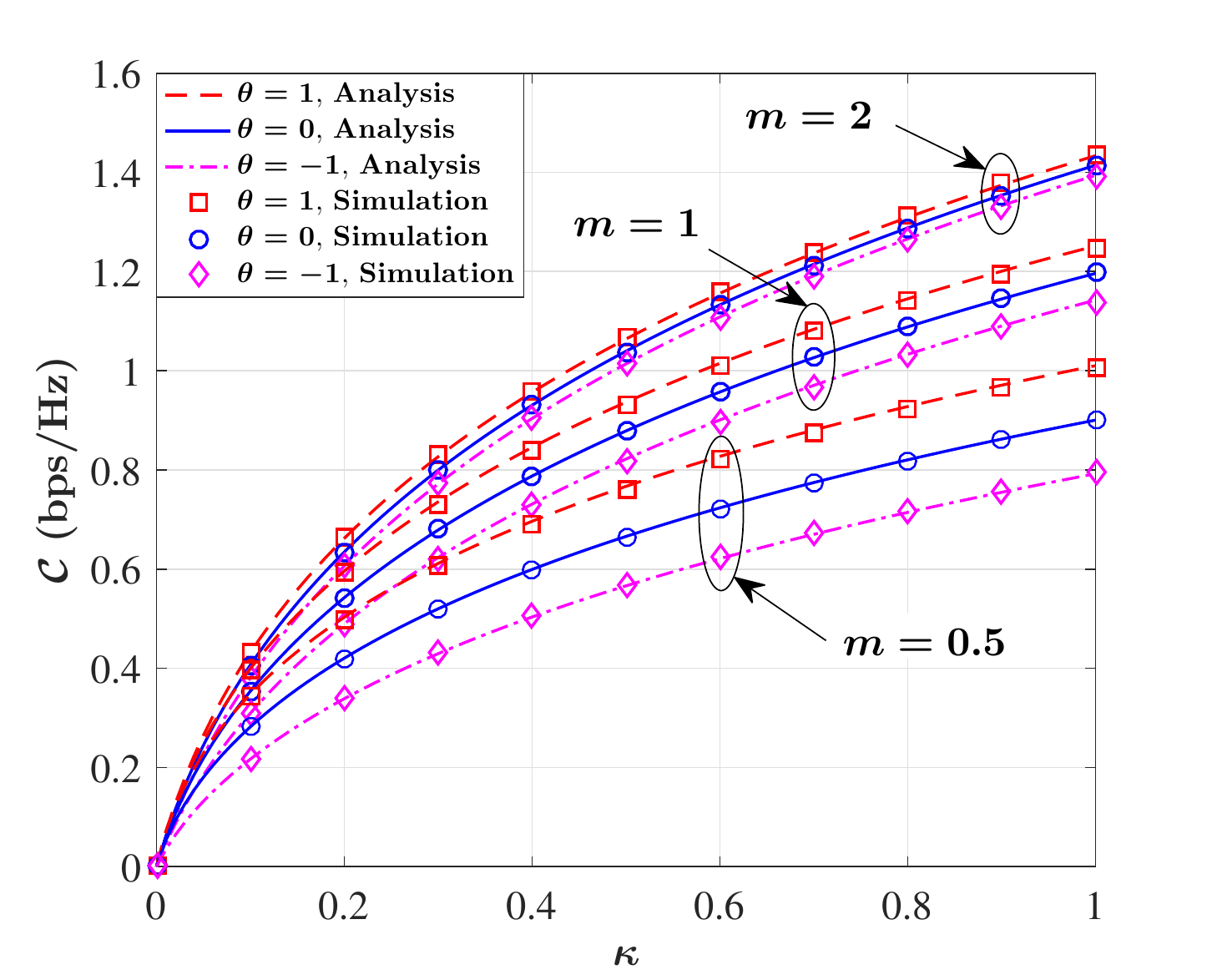} %\textwidth
		\caption{Ergodic capacity versus EH efficiency $\kappa$ under dependent/independent Nakagami-$m$ fading channels when $\rho=0.3$, $P_\mathrm{S}=10$W, $N=10^{-2}$W, $d_\mathrm{SR}=d_\mathrm{RD}=2$m, and $\alpha=2.5$.} %\vspace{-0.6cm}
		\label{fig-ckappam}
	\end{figure}%\vspace{0ex}
\begin{figure}[!h]\vspace{0ex}
	\centering
	\includegraphics[width=1\columnwidth]{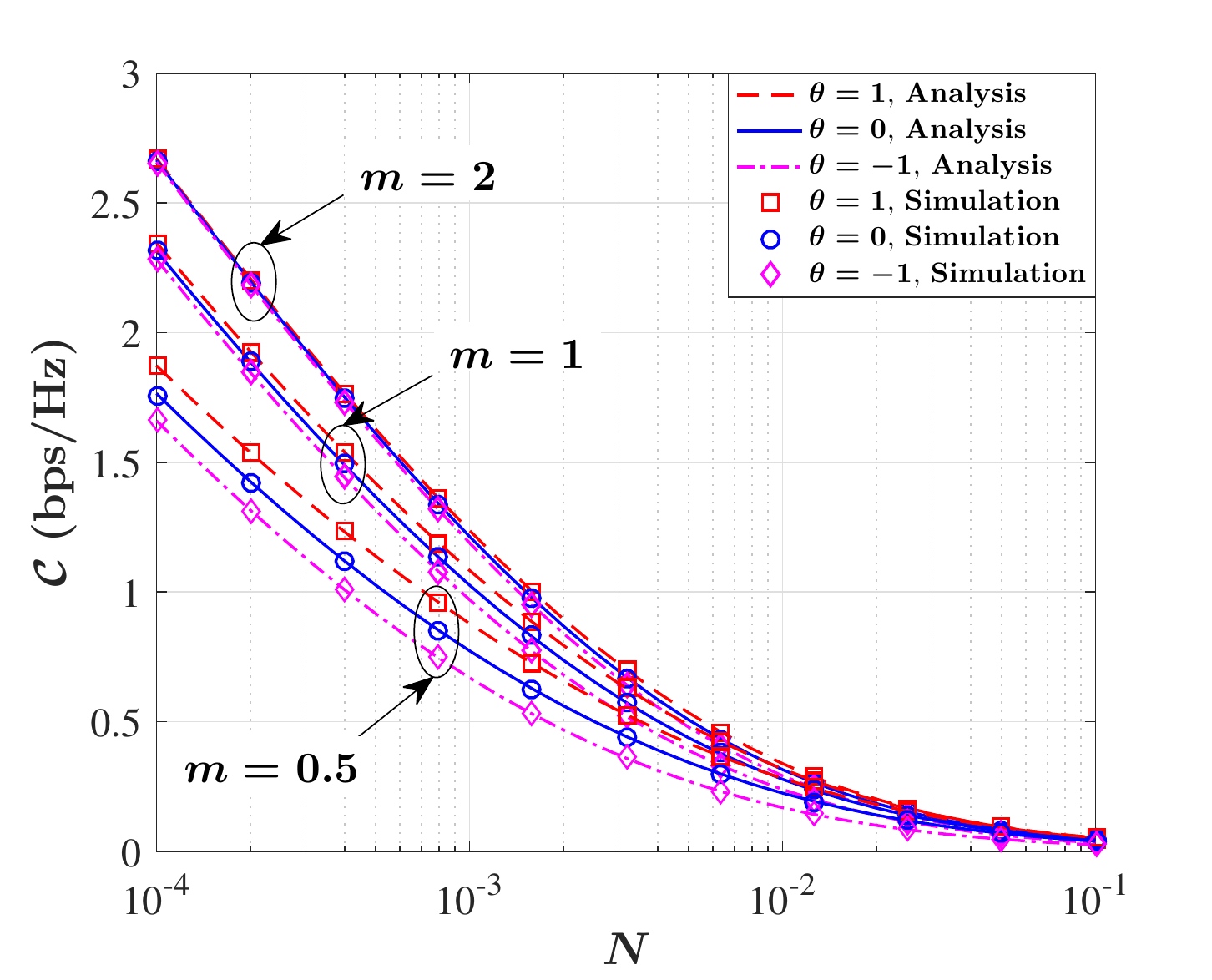} %\textwidth
	\caption{Ergodic capacity versus noise power $N$ under dependent/independent Nakagami-$m$ fading channels when $\kappa=0.7$, $\rho=0.3$, $P_\mathrm{S}=1$W, $d_\mathrm{RD}=d_\mathrm{SR}=2$m, and $\alpha=2.5$.} %\vspace{-0.6cm}
	\label{fig-cn}
\end{figure}%\vspace{0ex}
\begin{figure}[!h]\vspace{0ex}
	\centering
	\includegraphics[width=1\columnwidth]{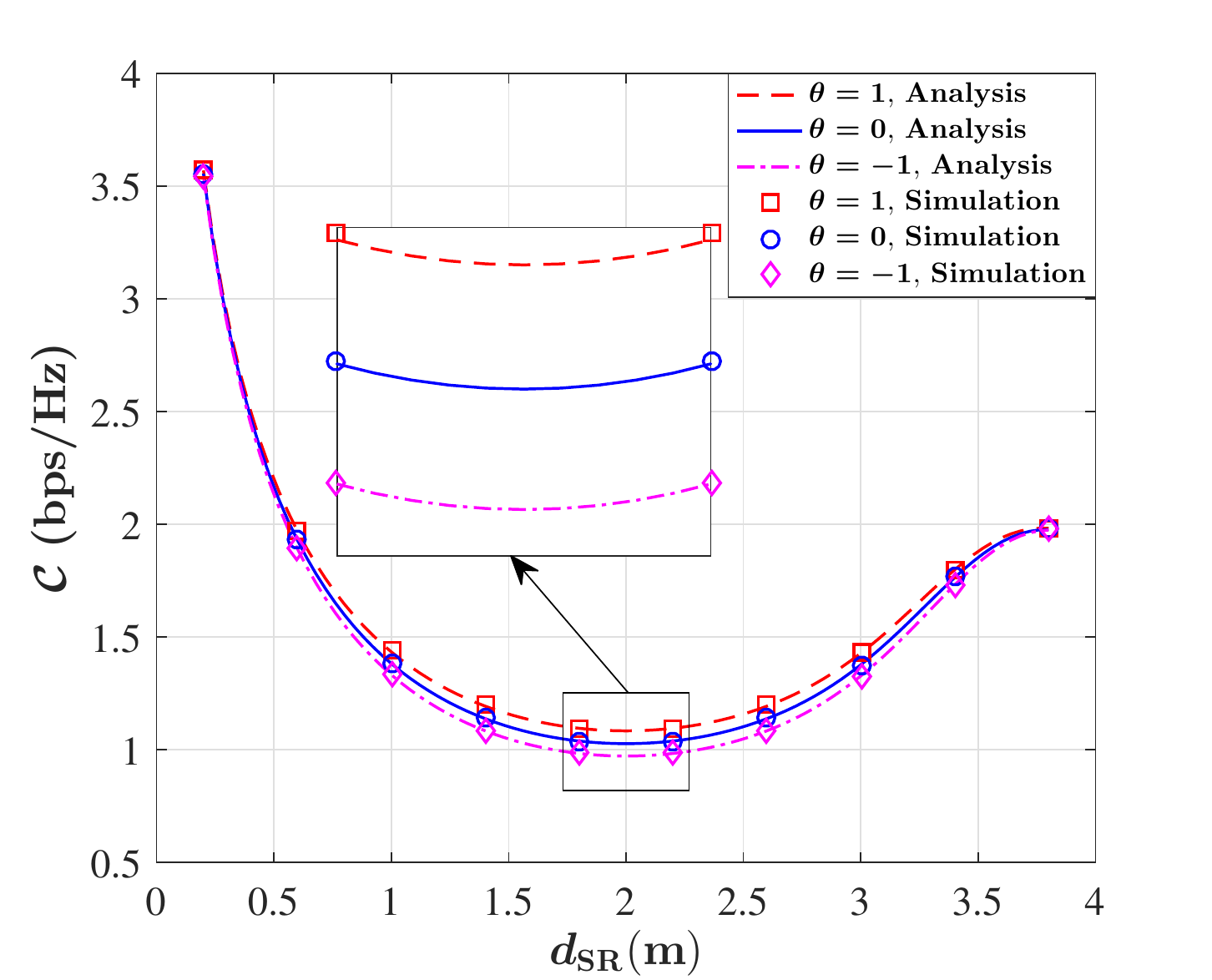} %\textwidth
	\caption{Ergodic capacity versus $SR$ distance $d_{SR}$ under dependent/independent Nakagami-$m$ fading channels when $\rho=0.3$, $m=1$, $P_\mathrm{S}=10$W, $N=10^{-2}$W, $d_\mathrm{RD}=4-d_\mathrm{SR}$, and $\alpha=2.5$.} %\vspace{-0.6cm}
	\label{fig-cdsr}
\end{figure}%\vspace{0ex}
	\begin{figure}[!h]\vspace{0ex}
		\centering
		\includegraphics[width=1\columnwidth]{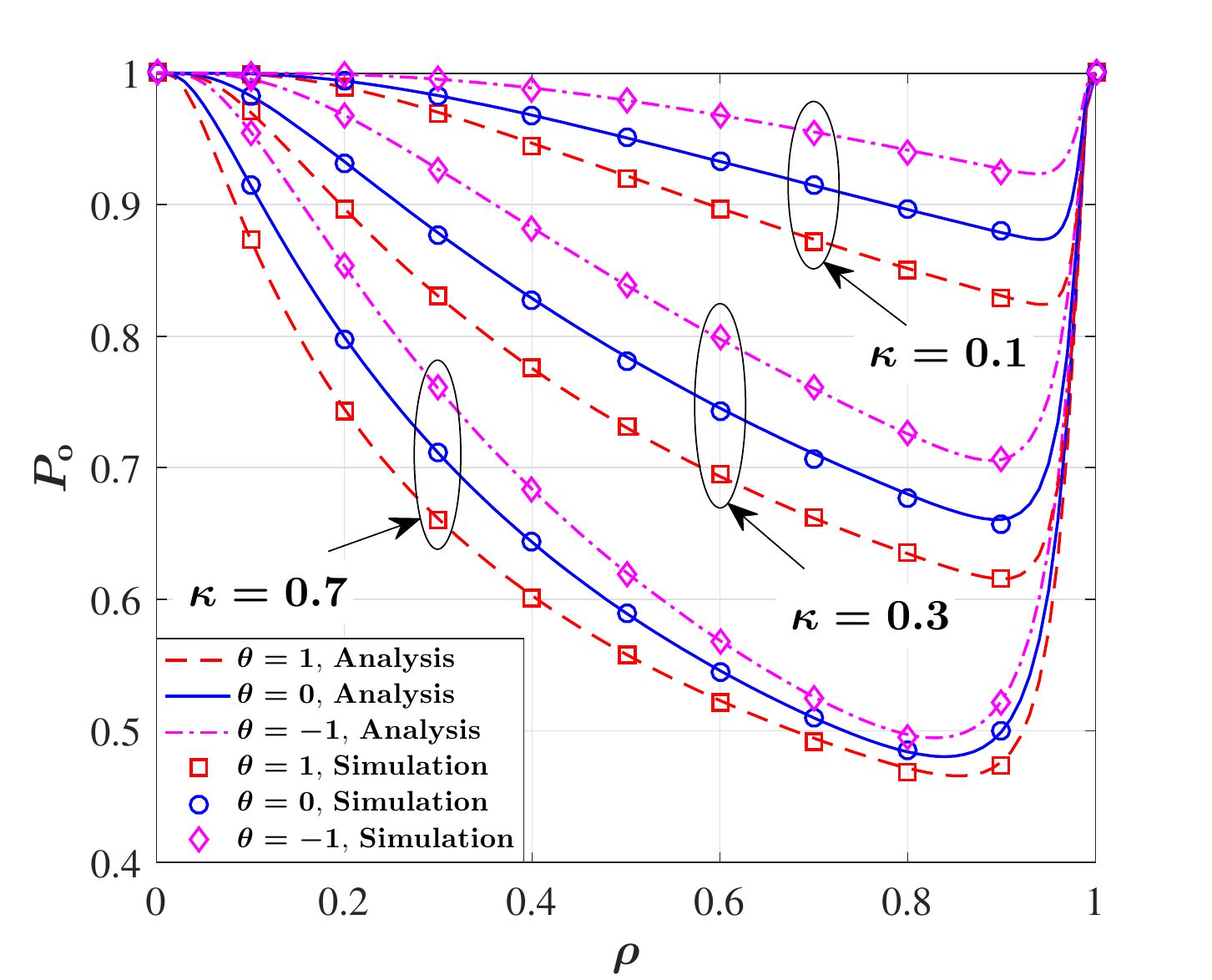} %\textwidth
		\caption{Outage probability versus PSR factor $\rho$ under dependent/independent Nakagami-$m$ fading channels when $m=1$, $P_\mathrm{S}=10$W, $N=10^{-3}$W, $d_\mathrm{SR}=d_\mathrm{RD}=2$m, $\gamma_\mathrm{t}=0$dB, and $\alpha=2.5$.} %\vspace{-0.6cm}
		\label{fig-outrhokappa2}
	\end{figure}%\vspace{0ex}
	\begin{figure}[!h]\vspace{0ex}
		\centering
		\includegraphics[width=1\columnwidth]{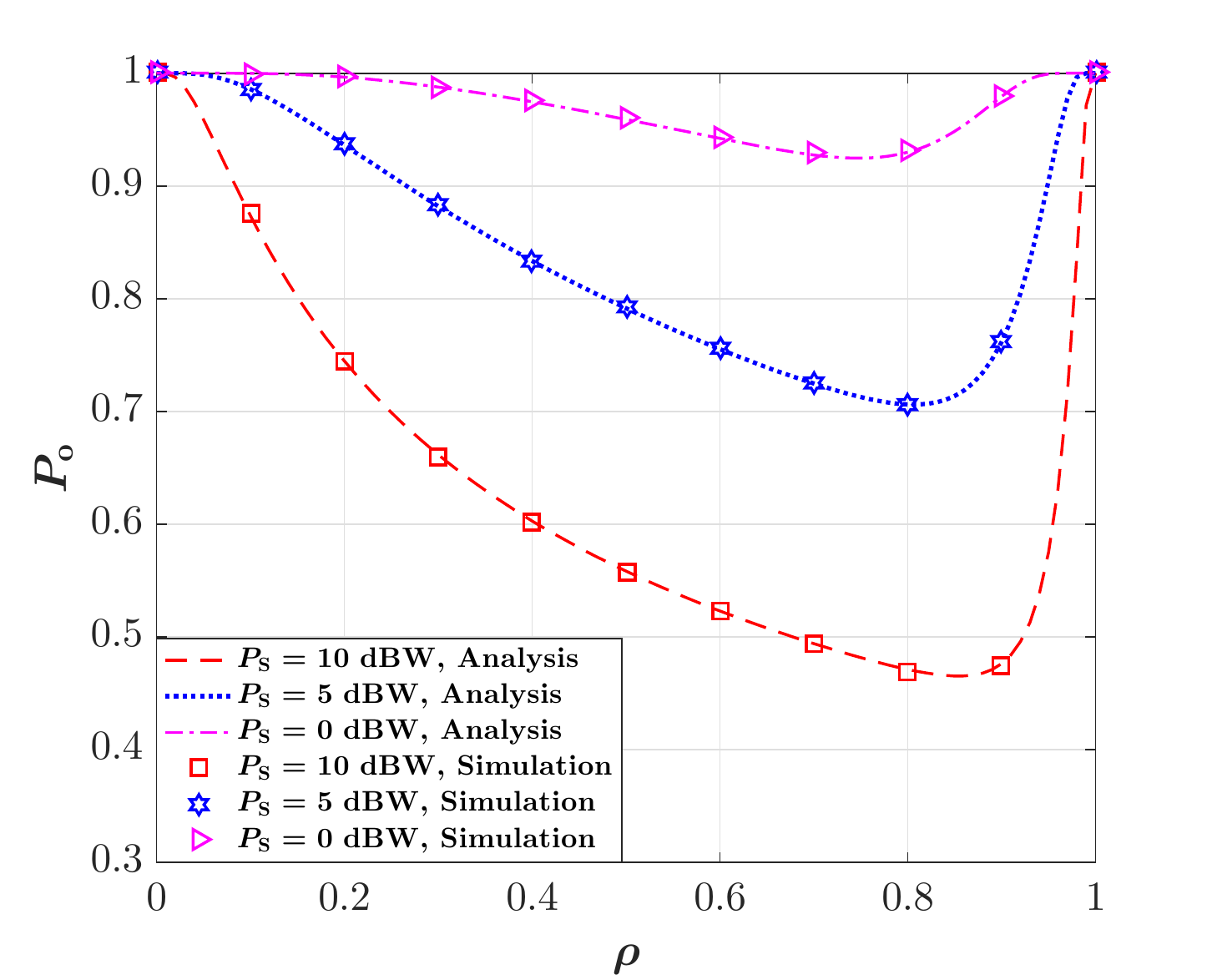} %\textwidth
		\caption{Outage probability versus PSR factor $\rho$ under dependent Nakagami-$m$ fading channels with positive structure when $\kappa=0.7$, $m=1$, $\gamma_\mathrm{t}=0$dB, $N=10^{-3}$W, $d_\mathrm{SR}=d_\mathrm{RD}=2$m, and $\alpha=2.5$.} %\vspace{-0.6cm}
		\label{fig-outrhops}
	\end{figure}%\vspace{0ex}
 	\begin{figure}[!h]\vspace{0ex}
		\centering
		\includegraphics[width=1\columnwidth]{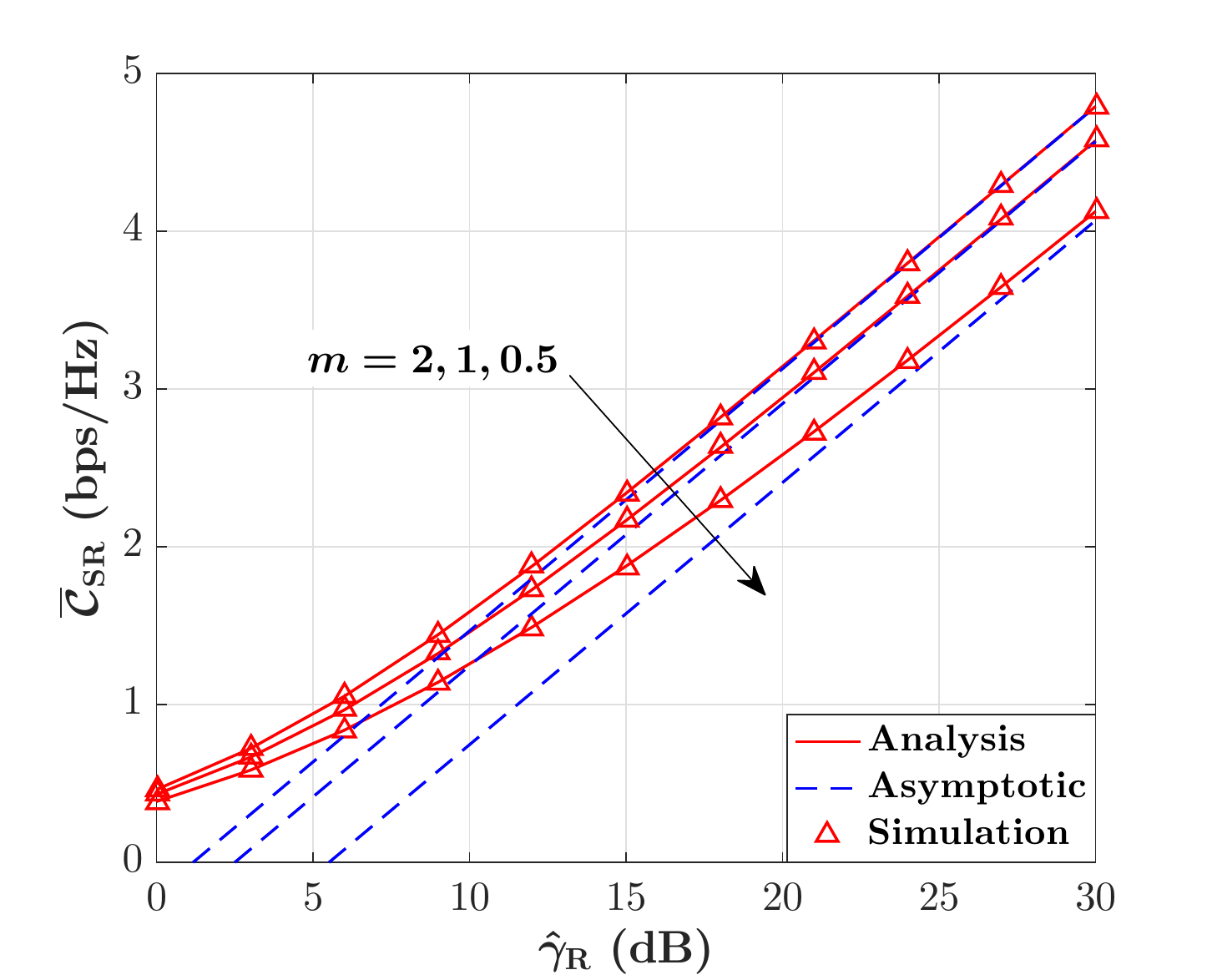} %\textwidth
		\caption{\textcolor{blue}{Average capacity of the SR link versus $\hat{\gamma}_\mathrm{R}$ for different values of fading parameter $m$, when $P_\mathrm{S}=10$W, $N=10^{-3}$W, $d_\mathrm{SR}=d_\mathrm{RD}=2$m, $\kappa=0.7$, $\rho=0.3$, and $\alpha=2.5$.}} %\vspace{-0.6cm}
		\label{fig-asym-csr}
	\end{figure}%\vspace{0ex}
 	\begin{figure}[!h]\vspace{0ex}
		\centering
		\includegraphics[width=1\columnwidth]{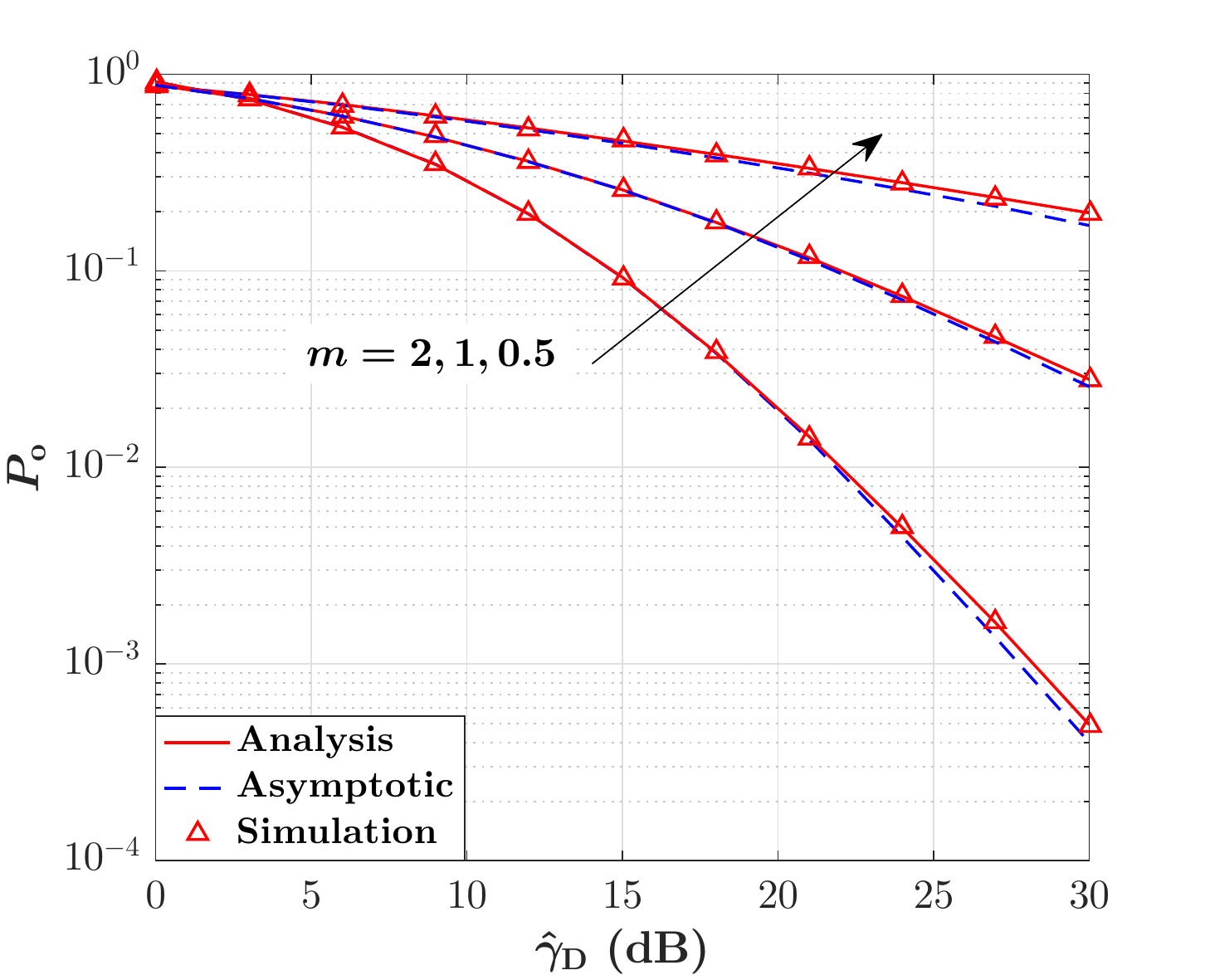} %\textwidth
		\caption{\textcolor{blue}{Outage probability versus ${\hat{\gamma}}_\mathrm{D}$ under positive correlation for different values of fading parameter $m$, when $P_\mathrm{S}=10$W, $N=10^{-3}$W, $d_\mathrm{SR}=d_\mathrm{RD}=2$m, $\kappa=0.7$, $\rho=0.3$, $\alpha=2.5$, $\gamma_\mathrm{t}=0$dB, and $\theta=1$.}} %\vspace{-0.6cm}
		\label{fig-asym-out}
	\end{figure}%\vspace{0ex}
	
	%\begin{figure}[!h]\vspace{0ex}
	%	\centering
	%	\includegraphics[width=1\columnwidth]{csnr.eps} %\textwidth
	%	\caption{$P_{out}$ versus $\overline{\gamma}_1$ over Fisher-Snedecor $\mathcal{F}$ fading doubly dirty MAC when $R_t=2.5$, $m_1=m_{1,s}=m_2=m_{2,s}=2$, and $\theta=40$, respectively.} %\vspace{-0.6cm}
	%	\label{}
	%\end{figure}%\vspace{0ex}
	
	\section{NUMERICAL RESULTS}\label{sec-results}
	In this section, we evaluate the theoretical expressions previously derived, which are double-checked in all instances with Monte-Carlo (MC) simulations\footnote{\textcolor{blue}{Basic simulation of copula-based dependence between RVs is natively implemented in software mathematical packages such as Matlab, and more sophisticated functionalities are available through recently developed packages \cite{Coblenz2021}. Besides, such dependence can also be easily implemented using state-of-the-art methods of transformation of RVs \cite{Joe2014}}.}.
	
	Fig. \ref{fig-crhom} illustrates the performance of ergodic capacity under correlated Nakagami-$m$ fading based on the variation of PSR factor $\rho$ for selected values of fading parameter $m$ and dependence parameter $\theta$. In both independent and correlated fading conditions, we can see that the capacity performance is improved as $\rho$ increases from $0$ to an optimal value and it is weakened as $\rho$ grows from the optimal value to $1$. The reason for this behavior is that as $\rho$ increases from $0$ to its optimal value, more power is allocated to the EH process, and thus, the relay node $R$ can transmit information with a higher power, which improves the capacity performance. However, as $\rho$ grows from its optimal amount to $1$, more power is dedicated to the EH process and less power remains for $SR$ information processing, so, the capacity value decreases. Furthermore, it can be seen that correlated fading provides better performance as compared with negative dependence structure as well as the independent case in terms of the ergodic capacity. We can also observe that as the fading severity reduces (i.e., $m$ increases), the ergodic capacity performance ameliorates but the fading correlation effects are gradually eliminated. The behavior of the ergodic capacity in terms of the source power $P_\mathrm{S}$ for the selected values of the fading and dependence parameters is shown in Fig. \ref{fig-cpsm}. We can see that under fixed values of the PSR factor $\rho$ and EH efficiency $\kappa$, the ergodic capacity performance improves as $P_\mathrm{S}$ and $m$ increase, as expected. Fig. \ref{fig-ckappam} represents the ergodic capacity performance with respect to the EH efficiency $\kappa$ for selected values of the fading and dependence parameters. It can be observed that a larger ergodic capacity is achieved when $\kappa$ tends to $1$ since an increment in EH efficiency leads more energy reach to the harvester in each slot. The ergodic capacity performance based on the variations of the noise power $N$ for given values of the fading and dependence parameters is illustrated in Fig. \ref{fig-cn}, where lower values of the ergodic capacity achieve as the noise power decreases. As expected, the ergodic capacity performance also improves as $m$ increases under both dependent and independent structures. Fig. \ref{fig-cdsr} indicates the behavior of the ergodic capacity in terms of the $SR$ distance $d_\mathrm{SR}$ under different values of the dependence parameter. In this case, the $RD$ distance $d_\mathrm{RD}$ is set to $d_\mathrm{RD}=4-d_\mathrm{SR}$, and other parameters are kept fixed. For $d_\mathrm{SR}<2$, it can be observed that the ergodic capacity decreases as the distance between the source node and the relay node increases. The reason is that by increasing $d_\mathrm{SR}$, both energy harvested and the received signal strength at the relay node decrease due to the larger path-loss $d_\mathrm{SR}^\alpha$. Therefore, the received signal strength at the destination node is poor and the ergodic capacity decreases. However, we can see that the ergodic capacity slightly increases by increasing $d_\mathrm{SR}$ beyond $2$. This is because as the relay node gets closer to the destination ($d_\mathrm{RD}<2$), even lesser values of harvested energy, the harvested energy is enough for reliable communication between the relay and the destination nodes due to smaller values of the $RD$ path-loss $d_\mathrm{RD}^{\alpha}$.
	
	Figs. \ref{fig-outrhokappa2} and \ref{fig-outrhops} show the behavior of the outage probability in terms of $\rho$ under correlated Nakagami-$m$ fading channels for different values of the EH efficiency $\kappa$, dependence parameter $\theta$, and source power $P_\mathrm{S}$. With the same argument adopted in analyzing Fig. \ref{fig-crhom}, we observe that there is a trade-off between energy harvesting and information transmission in terms of the outage probability, such that the minimum outage probability is achieved for an optimal value of $\rho$ under both correlated and independent fading scenarios. Furthermore, it is worth noting that the outage probability performance improves as the EH efficiency $\kappa$ grows since the harvester gains more energy in each available slot. We can also see that the outage probability improvement is increased as $P_\mathrm{S}$ rises under all dependence structures. From correlation viewpoint, it can be realized from both Figs. \ref{fig-outrhokappa2} and \ref{fig-outrhops} that the correlated fading provides a smaller outage probability under positive dependence structure compared with negative correlation and independent case.

   \textcolor{blue}{Fig. \ref{fig-asym-csr} shows the exact and asymptotic behavior of the average capacity for the SR link in terms of $\hat{\gamma}_\mathrm{R}$ . It can be seen that as $\hat{\gamma}_\mathrm{R}$ grows, a larger capacity is provided for the SR link. Furthermore, it becomes evident that the capacity highly depends on the fading parameter in the high SNR regime (see \eqref{csr-asy}) so that as $m$ increases the capacity performance improves. In Fig. \ref{fig-asym-out}, the exact and asymptotic performance of the outage probability in terms of the average SNR ${\hat{\gamma}_\mathrm{D}}$ under positive correlation is represented. It is confirmed that the asymptotic analysis reaches the exact analysis with high accuracy. In addition, we can see that as $\hat{\gamma}_\mathrm{D}$ increases, the outage probability performance improves since the channel conditions become better. Furthermore, we can also observe that in both exact and asymptotic analysis, the outage probability performance highly depends on the fading parameter, namely, lower outage probability is achieved as $m$ grows.}%we can see that in the high SNR regime (i.e., $\hat{\gamma}_\mathrm{R} \rightarrow \infty$), the effect of fading parameter $m$ is more noticeable compared with the exact expression provided in \eqref{out}.}
   
	\section{Conclusion}\label{sec-conclusion}
	In this paper, we analyzed the effect of fading correlation on the performance of SWIPT relay networks, where the power splitting relaying protocol is used for the energy harvesting process. To this end, we first provided general analytical expressions of the %PDF and 
	CDF for the product of two arbitrary dependent random variables. Then, we obtained the closed-form expressions for the ergodic capacity and outage probability under correlated Nakagami-$m$ fading channels, using FGM copula. \textcolor{blue}{Our theoretical results, confirmed by simulations, show that a positive dependence between the SR and RD links has a constructive effect on the system performance, while the converse conclusion is obtained for the case of negative dependence. These effects become more important as fading severity grows.}
 %The numerical and simulation results showed that considering correlated fading under the positive dependence structure has a constructive effect on the performance of ergodic capacity and outage probability, while negative correlation is detrimental for the system performance. It was also shown that as the fading severity decreases, the system performance improves under PSR  protocol. 

	%there is an optimal value for the PSR factor $\rho$ which maximizes the system capacity. 
	\appendices
	\section{Proof of Theorem \ref{thm-product-cdf-pdf}}\label{app-thm-product-cdf-pdf}
	By assuming $Y_1=X_1X_2$ and $Y_2=X_1$, and exploiting the PDF of $Y_1$ determined in \cite[Thm. 3]{ghadi2022capacity} as:
	%the joint PDF of $Y_1$ and $Y_2$ can be defined as:
%	\begin{align}
%		&f_{Y_1,Y_2}(y_1,y_2)=\Big|\frac{1}{y_2}\Big|f_{Y_2,\frac{Y_1}{Y_2}}\left(y_2,\frac{y_1}{y_2}\right)\\
%		&\overset{(a)}{=}\Big|\frac{1}{y_2}\Big|c\left(F_{Y_2}(y_2),F_{\frac{Y_1}{Y_2}}\left(\frac{y_1}{y_2}\right)\right)f_{Y_2}(y_2)f_{\frac{Y_1}{Y_2}}\left(\frac{y_1}{y_2}\right), 
%	\end{align}
%	where $(a)$ is followed by Corollary \ref{col-joint}. Now, using definition, the PDF of $Y_1$ can be determined as:
	\begin{align}
		f_{Y_1}(y_1)
		%\int_{-\infty}^{\infty}\Big|\frac{1}{y_2}\Big|c\left(F_{Y_2}(y_2),F_{\frac{Y_1}{Y_2}}\left(\frac{y_1}{y_2}\right)\right)\\
		%&\quad\quad\quad\quad\times f_{Y_2}(y_2)f_{\frac{Y_1}{Y_2}}\left(\frac{y_1}{y_2}\right)dy_2\\
		=\int_{0}^{1}c\left(u,F_{\frac{Y_1}{Y_2}}\left(\tfrac{y_1}{F_{Y_2}^{-1}(u)}\right)\right)
		\tfrac{f_{\frac{Y_1}{Y_2}}\left(\tfrac{y_1}{F_{Y_2}^{-1}(u)}\right)}{|F_{Y_2}^{-1}(u)|} du,
	\end{align}
	the CDF of $Y_1$ can be defined as:
	\begin{align}
		F_{Y_1}(t)=\int_{0}^{1}\int_{-\infty}^{t}c\left(u,F_{\frac{Y_1}{Y_2}}\left(\tfrac{y_1}{F_{Y_2}^{-1}(u)}\right)\right)
		\tfrac{f_{\frac{Y_1}{Y_2}}\left(\tfrac{y_1}{F_{Y_2}^{-1}(u)}\right)}{|F_{Y_2}^{-1}(u)|} dy_1du,
	\end{align}
	where $c(.)$ denotes the density of copula $C$. By taking change of variable  $v=F_{\frac{Y_1}{Y_2}}\left(\tfrac{y_1}{F_{Y_2}^{-1}(u)}\right)\Rightarrow dv=\tfrac{f_{\frac{Y_1}{Y_2}}\left(\tfrac{y_1}{F_{Y_2}^{-1}(u)}\right)}{F_{Y_2}^{-1}(u)}dy_1$, and since $F_{Y_2}^{-1}(u)\ge 0\Leftrightarrow u\ge 0$, and $F_{Y_2}^{-1}(u)\le 0\Leftrightarrow u\le 0$, we have
	\begin{align}\nonumber
		F_{Y_1}(t)=&-\int_{0}^{F_1(0)}\int_{1}^{F_{\frac{Y_1}{Y_2}}\left(\frac{t}{F_{Y_2}^{-1}(u)}\right)}\frac{\partial^2}{\partial u\partial v}C(u,v)dvdu\\
		&+\int_{F_1(0)}^{1}\int_{0}^{F_{\frac{Y_1}{Y_2}}\left(\frac{t}{F_{Y_2}^{-1}(u)}\right)}\frac{\partial^2}{\partial u\partial v}C(u,v)dvdu. 
	\end{align}
	Now, by computing the above integral, the proof is completed. The details of the proof can be obtained in \cite{ly2019determining}.
	\section{Proof of Theorem \ref{thm-cdf-pdf}}\label{app-thm-cdf-pdf}
	By applying the FGM copula to \eqref{def-cdf-d}, and then first derivation with respect to $F_{G_\mathrm{SR}}(G_\mathrm{SR})$, the CDF of $\gamma_\mathrm{D}$ can be rewritten as:
	\begin{align}\nonumber
		&F_{\gamma_\mathrm{D}}(\gamma_\mathrm{D})=\int_{0}^{\infty}f_{G_\mathrm{SR}}(g_\mathrm{SR})F_{G_\mathrm{RD}}\left(\frac{\gamma_\mathrm{D}}{\hat{\gamma}_\mathrm{D}g_\mathrm{SR}}\right)\\
		&\hspace{-1ex}\times\Bigg[1+\theta\left(1-F_{G_\mathrm{RD}}\left(\frac{\gamma_\mathrm{D}}{\hat{\gamma}_\mathrm{D}g_\mathrm{SR}}\right)\right)\left(1-2F_{G_\mathrm{SR}}\left(g_\mathrm{SR}\right)\right)\Bigg]dg_\mathrm{SR}\label{},\\
		&=1-\mathcal{I}_1+\theta\left[-\mathcal{I}_1+2\mathcal{I}_2+\mathcal{I}_3-2\mathcal{I}_4\right],\label{proof-cdf}
	\end{align}
	%\begin{align}\nonumber
	%F_y(Y)=&F_1(0)+\int_{0}^{\infty}\mathrm{sgn}(x_1)F_2\left(\frac{y}{x_1}\right)f_1\left(x_1\right)\\
	%&\times\Bigg[1+\theta\left(1-F_2\left(\frac{y}{x_1}\right)\right)\left(1-2F_1\left(x_1\right)\right)\Bigg]dx_1\\
	%&=I_1+\theta[I_1-2I_2-I_3+2I_4]\\
	%&=1-J_1+\theta\left[-J_1+2J_2+J_3-2J_4\right]
	%\end{align}
	where
	\begin{align}
		\mathcal{I}_1=\int_{0}^{\infty}f_{G_\mathrm{SR}}(g_\mathrm{SR})\overline{F}_{G_\mathrm{RD}}\left(\frac{\gamma_\mathrm{D}}{\hat{\gamma}_\mathrm{D}g_\mathrm{SR}}\right)dg_\mathrm{SR},
	\end{align}
	\begin{align}
		\mathcal{I}_2=\int_{0}^{\infty}f_{G_\mathrm{SR}}(g_\mathrm{SR})\overline{F}_{G_\mathrm{SR}}(g_\mathrm{SR})\overline{F}_{G_\mathrm{RD}}\left(\frac{\gamma_\mathrm{D}}{\hat{\gamma}_\mathrm{D}g_\mathrm{SR}}\right)dg_\mathrm{SR},
	\end{align}
	\begin{align}
		\mathcal{I}_3=\int_{0}^{\infty}f_{G_\mathrm{SR}}(g_\mathrm{SR})\left(\overline{F}_{G_\mathrm{RD}}\left(\frac{\gamma_\mathrm{D}}{\hat{\gamma}_\mathrm{D}g_\mathrm{SR}}\right)\right)^2dg_\mathrm{SR},
	\end{align}
	\begin{align}
		\mathcal{I}_4=\int_{0}^{\infty}f_{G_\mathrm{SR}}(g_\mathrm{SR})\overline{F}_{G_\mathrm{SR}}(g_\mathrm{SR})\left(\overline{F}_{G_\mathrm{RD}}\left(\frac{\gamma_\mathrm{D}}{\hat{\gamma}_\mathrm{D}g_\mathrm{SR}}\right)\right)^2dg_\mathrm{SR}.
	\end{align}
	Now, by inserting the marginal CDFs and PDFs of $g_\mathrm{SR}$ given in \eqref{pdf-g} and \eqref{cdf-g} to above integrals and exploiting the following integral format, i.e., 
	\begin{align}
		\int_0^{\infty } x^{\beta-1} \mathrm{e}^{- \left(\lambda x+\frac{\eta}{x}\right)} \, dx=2\,\eta^{\frac{\beta}{2}} \lambda^{-\frac{\beta}{2}} K_{-\beta}\left(2 \sqrt{\eta\lambda} \right),
	\end{align}
	the integrals $\mathcal{I}_w$ for $w\in\{1,2,3,4\}$ can be computed as:  
	\begin{align}
		&\mathcal{I}_1%\int_{0}^{\infty}f_1(x_1)\overline{F}_2\left(\frac{y}{x_1}\right)dx_1\\
		=\frac{{m}^{m}}{\Gamma(m)}\sum_{k=0}^{m-1}\frac{m^n\gamma_\mathrm{D}^n}{\hat{\gamma}_\mathrm{D}^n n!}\int_{0}^{\infty}g_\mathrm{SR}^{m-n-1}\mathrm{e}^{-mg_\mathrm{SR}-\frac{m\gamma_\mathrm{D}}{\hat{\gamma}_\mathrm{D}g_\mathrm{SR}}}dg_\mathrm{SR},\\
		&=\frac{{m}^{m}}{\Gamma(m)}\sum_{n=0}^{m-1}\frac{2m^n}{n!}   \left(\frac{\gamma_\mathrm{D}}{\hat{\gamma}_\mathrm{D}}\right)^{\frac{m+n}{2}} K_{n-m}\left(2m\sqrt{\frac{\gamma_\mathrm{D}}{\hat{\gamma}_\mathrm{D}}}\right),\label{i1}
	\end{align}
	\begin{align}\nonumber
		&\mathcal{I}_2%=\int_{0}^{\infty}f_1(x_1)\overline{F}_1(x_1)\overline{F}_2\left(\frac{y}{x_1}\right)dx_1\\
		=\frac{m^m}{\Gamma(m)}\sum_{k=0}^{m-1}\sum_{n=0}^{m-1}\frac{m^{k+n}\gamma_\mathrm{D}^n}{\hat{\gamma}_{\mathrm{D}}^nk!n!}\\
		&\quad\quad\times\int_{0}^{\infty}g_\mathrm{SR}^{k-n+m-1}\mathrm{e}^{-2mg_\mathrm{SR}-\frac{m\gamma_\mathrm{D}}{\hat{\gamma}_\mathrm{D}g_\mathrm{SR}}}dg_\mathrm{SR},\\\nonumber
		&=\frac{m^m}{\Gamma(m)}\sum_{k=0}^{m-1}\sum_{n=0}^{m-1}\frac{2^{\frac{-k-m+n+2}{2}}m^{k+n}}{\hat{\gamma}_\mathrm{D}^{\frac{k+n+m}{2}}k!n!}\\
		&\quad\times\gamma_\mathrm{D}^{\frac{k+n+m}{2}} K_{n-k-m}\left(2m\sqrt{\frac{2\gamma_\mathrm{D}}{\hat{\gamma}_\mathrm{D}}}\right),\label{i2}
	\end{align}
	\begin{align}\nonumber
		&\mathcal{I}_3%=\int_{0}^{\infty}f_1(x_1)\left(\overline{F}_2\left(\frac{y}{x_1}\right)\right)^2dx_1\\
		=\frac{m^m}{\Gamma(m)}\sum_{n=0}^{m-1}\sum_{l=0}^{m-1}\frac{1}{n!l!}\\
		&\quad\,\,\,\times\int_{0}^{\infty}g_\mathrm{SR}^{m-1}\mathrm{e}^{-mg_\mathrm{SR}-\frac{2m\gamma_\mathrm{D}}{\hat{\gamma}_\mathrm{D}g_\mathrm{SR}}}\left(\frac{m\gamma_\mathrm{D}}{\hat{\gamma}_\mathrm{D}g_\mathrm{SR}}\right)^{n+l}dg_\mathrm{SR}\\\nonumber
		&=\frac{m^m}{\Gamma(m)}\sum_{n=0}^{m-1}\sum_{l=0}^{m-1}\frac{2^{\frac{-l+m-n+2}{2}}m^{l+n}}{\hat{\gamma}_\mathrm{D}^{\frac{l+m+n}{2}}n!l!}\\
		&\quad\times\gamma_\mathrm{D}^{\frac{l+m+n}{2}} K_{l-m+n}\left(2m \sqrt{\frac{2\gamma_\mathrm{D}}{\hat{\gamma}_\mathrm{D}}}\right),\label{i3}
	\end{align}
	\begin{align}\nonumber
		\mathcal{I}_4%=\int_{0}^{\infty}f_1(x_1)\overline{F}_1(x_1)\left(\overline{F}_2\left(\frac{y}{x_1}\right)\right)^2dx_1\\
		&=\frac{{m}^{m}}{\Gamma(m)}\sum_{k=0}^{m-1}\sum_{n=0}^{m-1}\sum_{l=0}^{m-1}\frac{m^{k+n+l}\gamma_\mathrm{D}^{n+l}}{\hat{\gamma}_\mathrm{D}^{n+l}k!n!l!}\\
		&\quad\times\int_{0}^{\infty}g_\mathrm{SR}^{k-n-l+m-1}e^{-2mg_\mathrm{SR}-\frac{2m\gamma_\mathrm{D}}{\hat{\gamma}_\mathrm{D}g_\mathrm{SR}}}dg_\mathrm{SR}\\\nonumber
		&=\frac{{m}^{m}}{\Gamma(m)}\sum_{k=0}^{m-1}\sum_{n=0}^{m-1}\sum_{l=0}^{m-1}\frac{2 m^{k+n+l}}{\hat{\gamma}_\mathrm{D}^{\frac{k+n+l+m}{2}}k!n!l!}\\
		&\quad\times\gamma_\mathrm{D}^{\frac{k+n+l+m}{2}} K_{n+l-k-m}\left(4 m \sqrt{\frac{\gamma_\mathrm{D}}{\hat{\gamma}_\mathrm{D}}}\right).\label{i4}
	\end{align}
	Finally, by plugging \eqref{i1}-\eqref{i4} into \eqref{proof-cdf} and doing some algebraic simplifications, the proof is completed.

\bibliographystyle{IEEEtran}
\bibliography{sample.bib}
\end{document}